\documentclass[aps,prx,twocolumn,notitlepage,nofootinbib]{revtex4-1}
\usepackage{amsmath, amsthm, amssymb, mathtools}
\usepackage{enumitem}

\usepackage{ifpdf}
\ifpdf
\usepackage[pdftex]{graphicx}
\else
\usepackage[dvips]{graphicx}
\fi
\usepackage{tikz}
 	 \usetikzlibrary{arrows,backgrounds}
\usepackage[all]{xy}
\usepackage{tocvsec2}
\usepackage{comment}
\usepackage{bbm}
\usepackage{tabularx}
\usepackage{stmaryrd}

\input xy
\xyoption{all}

\usepackage[pdftex,plainpages=false,hypertexnames=false,colorlinks,pdfpagelabels]{hyperref}
\newcommand{\arxiv}[1]{\href{http://arxiv.org/abs/#1}{\tt arXiv:\nolinkurl{#1}}}
\newcommand{\arXiv}[1]{\href{http://arxiv.org/abs/#1}{\tt arXiv:\nolinkurl{#1}}}
\renewcommand{\doi}[1]{\href{http://dx.doi.org/#1}{{\tt DOI:#1}}}

\newcommand{\googlebooks}[1]{(preview at \href{http://books.google.com/books?id=#1}{google books})}

\usepackage{xcolor}
\definecolor{dark-red}{rgb}{0.7,0.25,0.25}
\definecolor{dark-blue}{rgb}{0.15,0.15,0.55}
\definecolor{medium-blue}{rgb}{0,0,.8}
\definecolor{violet}{RGB}{138,43,226}
\definecolor{DarkGreen}{RGB}{0,150,0}
\definecolor{rufous}{HTML}{a81c07}
\definecolor{Forestgreen}{rgb}{0.0, 0.5, 0.0}
\definecolor{darkaquamarine}{rgb}{.12, .47, .42}
\hypersetup{
   linkcolor={purple},
   citecolor={medium-blue}, urlcolor={medium-blue}
}

\usepackage{fullpage}

\theoremstyle{plain}
\newtheorem{thm}{Theorem}[section]
\newtheorem*{thm*}{Theorem}
\newtheorem{thmalpha}{Theorem}

\newtheorem{cor}[thm]{Corollary}

\newtheorem*{cor*}{Corollary}

\newtheorem*{conj*}{Conjecture}
\newtheorem{lem}[thm]{Lemma}

\newtheorem*{quest*}{Question}
\newtheorem*{claim*}{Claim}

\theoremstyle{definition}
\newtheorem{defn}[thm]{Definition}
\newtheorem{construction}[thm]{Construction}

\newtheorem{nota}[thm]{Notation}

\newtheorem{example}[thm]{Example}

\newtheorem{sub-ex}[thm]{Sub-Example}
\newtheorem{rem}[thm]{Remark}

\newtheorem*{rem*}{Remark}

\DeclareMathOperator{\Ad}{Ad}

\DeclareMathOperator{\End}{End}

\DeclareMathOperator{\Id}{\Id}

\DeclareMathOperator{\im}{im}
\DeclareMathOperator{\Irr}{Irr}

\DeclareMathOperator{\spt}{SPT}
\DeclareMathOperator{\qd}{QD}

\newcommand{\set}[2]{\left\{#1 \middle| #2\right\}}

\def\semicolon{;}
\def\applytolist#1{
    \expandafter\def\csname multi#1\endcsname##1{
        \def\multiack{##1}\ifx\multiack\semicolon
            \def\next{\relax}
        \else
            \csname #1\endcsname{##1}
            \def\next{\csname multi#1\endcsname}
        \fi
        \next}
    \csname multi#1\endcsname}

\def\calc#1{\expandafter\def\csname c#1\endcsname{{\mathcal #1}}}
\applytolist{calc}QWERTYUIOPLKJHGFDSAZXCVBNM;
\def\bbc#1{\expandafter\def\csname bb#1\endcsname{{\mathbb #1}}}
\applytolist{bbc}QWERTYUIOPLKJHGFDSAZXCVBNM;
\def\bfc#1{\expandafter\def\csname bf#1\endcsname{{\mathbf #1}}}
\applytolist{bfc}QWERTYUIOPLKJHGFDSAZXCVBNM;
\def\sfc#1{\expandafter\def\csname s#1\endcsname{{\sf #1}}}
\applytolist{sfc}QWERTYUIOPLKJHGFDSAZXCVBNM;
\def\fc#1{\expandafter\def\csname f#1\endcsname{{\mathfrak #1}}}
\applytolist{fc}QWERTYUIOPLKJHGFDSAZXCVBNM;
\def\rmc#1{\expandafter\def\csname rm#1\endcsname{{\mathrm #1}}}
\applytolist{rmc}QWERTYUIOPLKJHGFDSAZXCVBNM;

\newcommand{\Fun}{{\sf Fun}}
\newcommand{\Rep}{{\sf Rep}}

\newcommand{\Mod}{{\sf Mod}}

\newcommand{\Bim}{{\sf Bim}}

\newcommand{\Hilb}{{\sf Hilb}}

\newcommand{\noshow}[1]{}
\newcommand{\MR}[1]{}

\newcommand{\calH}{\mathcal H}
\newcommand{\C}{\mathbb C}

\newcommand{\Z}{\mathbb{Z}}

\usetikzlibrary{shapes}
\usetikzlibrary{cd}
\usetikzlibrary{backgrounds}
\usetikzlibrary{decorations,decorations.pathreplacing,decorations.markings,decorations.pathmorphing}
\usetikzlibrary{fit,calc,through}
\usetikzlibrary{external}
\usetikzlibrary{arrows}
\tikzset{vertex/.style = {shape=circle,draw,fill=black,inner sep=0pt,minimum size=5pt}}
\tikzset{edge/.style = {->,> = latex', bend right}}
\tikzset{
	super thick/.style={line width=3pt}
}
\tikzstyle{knot}=[preaction={super thick, white, draw}]

\tikzset{
    quadruple/.style args={[#1] in [#2] in [#3] in [#4]}{
        #1,preaction={preaction={preaction={draw,#4},draw,#3}, draw,#2}
    }
}
\tikzstyle{shaded}=[fill=red!10!blue!20!gray!30!white]
\tikzstyle{unshaded}=[fill=white]
\tikzstyle{empty box}=[circle, draw, thick, fill=white, opaque, inner sep=2mm]
\tikzstyle{annular}=[scale=.7, inner sep=1mm, baseline]
\tikzstyle{rectangular}=[scale=.75, inner sep=1mm, baseline=-.1cm]
\tikzstyle{mid>}=[decoration={markings, mark=at position 0.53 with {\arrow{>}}}, postaction={decorate}]
\tikzstyle{mid<}=[decoration={markings, mark=at position 0.5 with {\arrow{<}}}, postaction={decorate}]
\tikzstyle{over}=[double, draw=white, super thick, double=]
\tikzstyle{box} = [rectangle,draw,rounded corners=5pt,very thick]

\newcommand{\roundNbox}[6]{
	\draw[rounded corners=5pt, very thick, #1] ($#2+(-#3,-#3)+(-#4,0)$) rectangle ($#2+(#3,#3)+(#5,0)$);
	\coordinate (ZZa) at ($#2+(-#4,0)$);
	\coordinate (ZZb) at ($#2+(#5,0)$);
	\node at ($1/2*(ZZa)+1/2*(ZZb)$) {#6};
}

\newcommand{\tikzmath}[2][]{\vcenter{\hbox{\begin{tikzpicture}[#1]#2
\end{tikzpicture}}}
}

\begin{document}


\title{Boundary algebras of the Kitaev Quantum Double model}

\author{Mario Tomba$^{*1}$}
\author{Shuqi Wei$^{*2}$}
\author{Brett Hungar$^3$}
\author{Daniel Wallick$^3$}
\author{Kyle Kawagoe$^3$}
\author{Chian Yeong Chuah$^3$}
\author{David Penneys$^3$}
\affiliation{$^*$Both of these co-first authors contributed equally to this article.}
\affiliation{$^1$Department of Mathematics, Dartmouth College, Hanover, NH 03755, USA}
\affiliation{$^2$Department of Mathematics, University of California Berkeley, Berkeley, CA 94720, USA}
\affiliation{$^3$Department of Mathematics, The Ohio State University, Columbus, OH 43210, USA}

\begin{abstract}
The recent article [arXiv:2307.12552]
gave local topological order (LTO) axioms for a quantum spin system, showed they held in Kitaev's Toric Code and in Levin-Wen string net models, and gave a bulk boundary correspondence to describe bulk excitations in terms of the boundary net of algebras.
In this article, we prove the LTO axioms for Kitaev's Quantum Double model for a finite group $G$.
We identify the boundary nets of algebras with fusion categorical nets associated to $(\mathsf{Hilb}(G),\mathbb{C}[G])$
or 
$(\mathsf{Rep}(G),\mathbb{C}^G)$ 
depending on whether the boundary cut is rough or smooth respectively.
This allows us to make connections to work of Ogata on the type of the cone von Neumann algebras in the algebraic quantum field theory approach to topological superselection sectors.
We show that the boundary algebras can also be calculated from a trivial $G$-symmetry protected topological phase ($G$-SPT), and that the gauging map preserves the boundary algebras.
Finally, we compute the boundary algebras for the (3+1)D Quantum Double model associated to an abelian group.

\end{abstract}

\maketitle

\section{Introduction}

Long-range topological order has been studied extensively over the past few decades \cite{1610.03911}. 
However, central questions about the identification of topological order in a given microscopic model remain outstanding. 
By using the recent formalism of \cite{2307.12552}, we identify the topological order of the Quantum Double model in both two and three dimensions by computing their nets of boundary algebras.

One of the foundational principles of topological physics is the bulk-boundary correspondence \cite{PhysRevB.43.11025,PhysRevLett.95.226801, 1310.5708,PhysRevB.91.125124,1510.08744}.  A major consequence of this correspondence is that the bulk phase of matter can be determined from the boundary theory. For example, the boundary of a (2+1)D topological quantum field theory may host a (1+1)D conformal field theory (CFT) on its boundary \cite{1810.05697}. 
This relationship may be used to at least partially identify the bulk theory from the CFT, for example, through entanglement entropy \cite{1103.5437,MR3714128,MR3955452}. 

Identifying a topological phase from its boundary algebra has several advantages.  
First, as we demonstrate, this procedure does not require any special knowledge of the string operators of the theory.  
In comparison, although \cite{PhysRevB.101.115113} did not require knowledge of the string operators to identify the phase of matter in a microscopic model, they did assume the possibility of finding some specific local operators. 
The article \cite{MR4109024} gives another method to determine modular data of the bulk theory from a ground state satisfying an area law, but at this time, the associators/F-symbols remain out of reach. Both of these approaches operate in the bulk.

As an aside, we expect that the boundary algebra method is amenable to numerical studies.  
Given a Hamiltonian on a finite lattice it should be possible to gain enough numerical information about the boundary algebra to identify it in full, at least on this finite lattice.

The recent article \cite{2307.12552} gives mathematical axioms for a (2+1)D quantum system to be \emph{locally topologically ordered} (LTO) building on the topological quantum order (TQO) axioms of \cite{MR2742836}.
Given such a spin system and a choice of half-plane bounded by a 1D sub-lattice,  one associates a 1D net of boundary algebras $I\mapsto\fB(I)$ along the sublattice.
It is conjectured that the \emph{DHR bimodules} \cite{2304.00068} of the boundary algebra describe the localized anyonic excitations of the bulk theory, giving a bulk-boundary correspondence.

This corrspondence has been verified for Kitaev's Toric Code \cite{MR1611329} and the Levin-Wen string net model \cite{PhysRevB.71.045110,PhysRevB.103.195155}.
For both these models, the boundary net of algebras $I\mapsto \fB(I)$ is described by a \emph{fusion categorical net} \cite{2307.12552}.
In more detail, given a unitary fusion category (UFC) $\cC$, we let $X \coloneqq \bigoplus_{c \in \Irr(\cC)} c$ be the direct sum over all simples.
The fusion categorical net is given by
$$
I\mapsto \End_\cC(X^{|I|}),
$$
where $|I|$ is the number of sites in the 1D interval $I$.
It was shown in \cite{2304.00068} that the DHR bimodules of this fusion categorical net is braided equivalent to $\cZ(\cC)$, the Drinfeld center, which is well known to describe the anyonic excitations.
In the case of the Toric Code, $\cC=\Hilb(\bbZ/2)$, giving the double $\cD(\bbZ/2)$ as the anyonic excitations \cite{MR1951039}.  

In this article, we verify the LTO axioms for Kitaev's Quantum Double model  \cite{MR1951039}, and we verify the bulk-boundary correspondence by showing the nets of boundary algebras are again fusion categorical nets.

\begin{thmalpha}
\label{thm:BoundaryAlgebrasForKQD}
The Kitaev Quantum Double model for a finite group $G$ satisfies the LTO axioms of \cite{2307.12552}.
The boundary algebra for a rough cut is the fusion categorical net for $\Hilb(G)$ with generator $\bbC[G]$, and the boundary algebra for a smooth cut is the fusion categorical net for $\Rep(G)$ with generator $\bbC^G$.
\end{thmalpha}

The UFCs $\Hilb(G)$ or $\Rep(G)$ are equivalent when $G$ is abelian, but not when $G$ is nonabelian.  
However, both nets have the same category of DHR bimodules as $\cZ(\Hilb(G)) \cong \cZ(\Rep(G))$, which is the well-known category of excitations for Kitaev's Quantum Double model \cite{MR1951039}. 

While the first part of this theorem was completely expected, the second part is somewhat surprising in light of \cite[Rem.~4.10]{2307.12552}.
There, boundary algebras for the version of the Levin-Wen model from \cite{MR3204497,2305.14068} in which degrees of freedom are located on vertices were computed.
It was shown that Levin-Wen models for $\Hilb(S_3)$ and $\Rep(S_3)$ give highly different limit AF $\rmC^*$-algebras $\varinjlim_I \fB(I)$ that are not stably isomorphic.
However, for the Kitaev Quantum Double model, the limit AF $\rmC^*$-algebras are both the same UHF (ultra hyperfinite) algebra $M_{|G|^\infty}$.

In contrast to \cite[Thm.~B]{2307.12552}, we have the following somewhat surprising result.

\begin{thmalpha}
\label{thm:CanonicalStateTracialIntro}
The canonical state $\psi$ restricted to the boundary algebras of the Kitaev Quantum Double model is always tracial.
\end{thmalpha}

For the version of the Levin-Wen model studied in \cite{2307.12552}, the canonical state gives a trace on the boundary algebra if and only if the UFC is \emph{pointed}, i.e., every simple object is invertible ($d_c=1$ for all $c\in \Irr(\cC)$). 
Thus the canonical state for the Levin-Wen model for $\Rep(G)$ is not tracial when $G$ is non-abelian, in contrast to Theorem \ref{thm:CanonicalStateTracialIntro} above.

Moreover, Theorem \ref{thm:CanonicalStateTracialIntro}
gives a corollary concerning cone algebras for the Kitaev Quantum Double model, which have been used to study the superselection theory in the algebraic quantum field theoretic sense \cite{MR3426207}.  
Recently, Ogata proved that the cone algebras for the Quantum Double model are type $\rmI\rmI_\infty$ factors \cite{2212.09036}.  
Our theorem provides an independent proof of this result.  

We also compute the boundary algebra of the Quantum Double by recognizing that it is a non-twisted bosonic lattice gauge theory \cite{RevModPhys.51.659,PhysRevB.87.125114}. 
In particular, it is the gauge theory of a trivial $G$-Symmetry Protected Topological phase ($G$-SPT). 
The gauging map provides a dictionary between the low energy operators of the two theories.  
The utility of this fact for our purposes is expressed in the following theorem.

\begin{thmalpha}
    For a finite group $G$, the boundary algebra of a bosonic $G$-SPT is isomorphic to that of the corresponding gauge theory.  
    In particular, the boundary algebra of the trivial $G$-SPT is isomorphic to the boundary algebra of Kitaev's Quantum Double.
\end{thmalpha}

Unlike long range topological orders, it is known how to uniquely identify the bulk bosonic SPT order from the boundary theory in (2+1)D \cite{PhysRevB.90.235137,PhysRevB.104.115156}. 
In fact, it is known how to compute the bulk order from certain operators in the boundary algebra specifically \cite{PhysRevB.104.115156}. 
It is therefore reasonable to anticipate our result: one can identify the bulk phase of matter of the Quantum Double, a dual theory of the trivial SPT, from the boundary algebra of the SPT.

Finally, in order to demonstrate how to apply our methods in three spatial dimensions, we compute the boundary algebra of the 3D Quantum Double where the underlying group is abelian. 
We believe that our arguments may be extended to any dimension.

\subsection*{Acknowledgements}
This article is the undergraduate research project of Mario Tomba and Shuqi Wei, who were supervised by Brett Hungar, Daniel Wallick, Chian Yeong Chuah, Kyle Kawagoe, and David Penneys during Summer 2023,
supported by Penneys’ NSF grants
DMS CAREER 1654159 and DMS 2154389.
The authors would like to thank
Corey Jones,
Junhwi Lim,
and Pieter Naaijkens
for helpful conversations.

\section{Local topological order axioms and boundary algebras}
\label{sec:LTO}

In this section, we give a brief overview of the local topological order axioms given in \cite{2307.12552}, simplified to the $(2+1)D$ spin system setting.  
For more detail we refer to reader to Sections 2.1 and 2.2 of that paper.  

Let $\cL \subseteq \bbR^2$ be a 2D lattice, where each site $\ell\in\cL$ carries $\bbC^d$ spins.
For simplicity, we assume $\cL$ is either $\bbZ^2\subseteq \bbR^2$ or the 2D edge lattice.
The \emph{quasi-local algebra} $\fA$ is the UHF (uniformly hyperfinite) infinite tensor product $\rm C^*$-algebra $\bigotimes_{\ell\in\cL} M_d(\bbC)$.
For a contractible bounded region $\Lambda \subseteq \cL$, which we usually assume to be a rectangle, we define the \emph{local operators} $\fA(\Lambda)\coloneqq\bigotimes_{\ell\in\Lambda}M_d(\bbC)$.
The assignment $\Lambda\mapsto \fA(\Lambda)$ forms a \emph{net of algebras} satisfying the following axioms:
\begin{itemize}
\item 
$\fA(\emptyset) = \bbC 1_{\fA}$,
\item 
If $\Lambda \subseteq \Delta$, then $\fA(\Lambda) \subseteq \fA(\Delta)$,
\item 
If $\Lambda_1 \cap \Lambda_2 = \emptyset$, then $[\fA(\Lambda_1), \fA(\Lambda_2)] = 0$,
\item 
\(\overline{\bigcup \fA(\Lambda)}^{\| \cdot \|} = \fA\).
\end{itemize}

Note that there is a canonical action of $\bbZ^2$ on $\fA$ by translation.
For $g \in \bbZ^2$, $g \cdot \fA(\Lambda) = \fA(g + \Lambda)$.

\begin{defn}[Net of projections]
A \emph{net of projections} is an assignment of an orthogonal projector $p_\Lambda\in \fA(\Lambda)$ satisfying
$p_\Delta \leq p_\Lambda$ ($p_\Delta p_\Lambda p_\Delta = p_\Delta$) if $\Lambda \subseteq \Delta$.  

All nets of projections in this article are assumed to be \emph{translation invariant}, i.e., 
$g \cdot p_\Delta = p_{g + \Delta}$ for all $g \in \bbZ^2$.  
\end{defn}

\begin{rem}
A net of projections is only required to assign projections to rectangles which are \emph{sufficiently large}, meaning that there is a global constant $r > 0$ such that $\Lambda$ contains an $r \times r$ square.  
If a rectangle $\Lambda$ is not sufficiently large, we set $p_\Lambda \coloneqq 1$, which extends $p$ to a well-defined net of projections on all rectangles.

When $\cL$ is the edge lattice, we count $r$ by projecting sites down to the $1D$ slices in the $x$ and $y$ direction.
For example, the following rectangles are sufficiently large for the specified $r$ values below.
$$
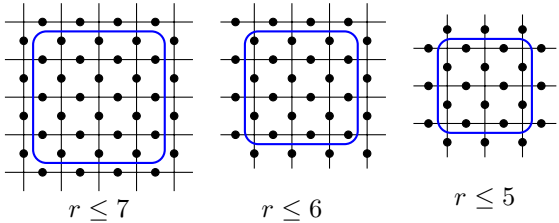
\tikzmath{
\draw[scale=.5] (-2.5,-2.5) grid (2.5,2.5);
\foreach \x in {-.75,-.25,.25,.75}{
\foreach \y in {-1,-.5,...,1}{
\filldraw (\x,\y) circle (.05cm);
}}
\foreach  \x in {-1,-.5,...,1}{
\foreach \y in {-.75,-.25,.25,.75}{
\filldraw (\x,\y) circle (.05cm);
}}
\draw[thick, blue, rounded corners=5pt] (-.875,-.875) rectangle (.875,.875);
\node at (0, -1.5) {$r\leq 7$};
}
\quad
\tikzmath{
\draw[scale=.5] (-1.9,-1.9) grid (2.5,2.5);
\foreach \x in {-.75,-.25,.25,.75}{
\foreach \y in {-.5,0,.5,1}{
\filldraw (\x,\y) circle (.05cm);
}}
\foreach \x in {-.5,0,.5,1}{
\foreach \y in {-.75,-.25,.25,.75}{
\filldraw (\x,\y) circle (.05cm);
}}
\draw[thick, blue, rounded corners=5pt] (-.625,-.625) rectangle (.875,.875);
\node at (0, -1.5) {$r\leq 6$};
}
\quad
\tikzmath{
\draw[scale=.5] (-1.9,-1.9) grid (1.9,1.9);
\foreach \x in {-.75,-.25,.25,.75}{
\foreach \y in {-.5,0,.5}{
\filldraw (\x,\y) circle (.05cm);
}}
\foreach \x in {-.5,0,.5}{
\foreach \y in {-.75,-.25,.25,.75}{
\filldraw (\x,\y) circle (.05cm);
}}
\draw[thick, blue, rounded corners=5pt] (-.625,-.625) rectangle (.625,.625);
\node at (0, -1.5) {$r\leq 5$};
}
$$

For our example, we work with the edge lattice and $r = 3$.  
\end{rem}

\begin{defn}[Surrounding regions] \label{defn:surrounding-regions}
Let $\Lambda$ and $\Delta$ be rectangles in $\cL$ with $\Lambda \subset \Delta$.
Fix a \emph{surrounding constant} $s>0$.
We say: 
\begin{itemize}
\item 
$\Lambda \subset_s \Delta$ if every lattice point $\ell\in \Delta\setminus \Lambda$ is contained in an $s\times s$ rectangle contained entirely in $\Delta\setminus \Lambda$.
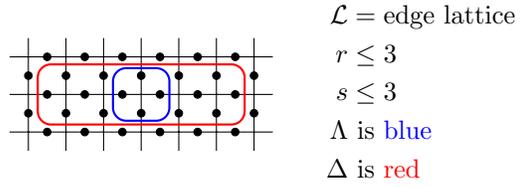
$$
\tikzmath{
\draw[scale=.5] (-3.5,-1.5) grid (3.5,1.5);
\foreach \x in {-1.25,-.75,...,1.25}{
\foreach \y in {-.5,0,.5}{
\filldraw (\x,\y) circle (.05cm);
}}
\foreach \x in {-1.5,-1,...,1.5}{
\foreach \y in {-.25,.25}{
\filldraw (\x,\y) circle (.05cm);
}}
\draw[thick, blue, rounded corners=5pt] (-.375,-.35) rectangle (.375,.35);
\draw[thick, red, rounded corners=5pt] (-1.375,-.4) rectangle (1.375,.4);
}
\qquad
\begin{aligned}
\cL&=\text{edge lattice}
\\
r&\leq 3
\\
s&\leq 3
\\
\Lambda &\text{ is \textcolor{blue}{blue}}
\\
\Delta &\text{ is \textcolor{red}{red}}
\end{aligned}
$$
In the cartoon above, $\partial \Lambda \cap \partial \Delta$ consists of two disjoint intervals on the top and bottom, each consisting of one point.
\item
$\Lambda$ is \emph{completely $s$-surrounded} by $\Delta$, denoted $\Lambda \ll_s \Delta$, if $\Lambda \subset_s \Delta$ and $\partial \Lambda \cap \partial \Delta = \emptyset$.
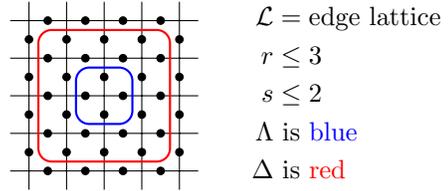
$$
\tikzmath{
\draw[scale=.5] (-2.5,-2.5) grid (2.5,2.5);
\foreach \x in {-.75,-.25,.25,.75}{
\foreach \y in {-1,-.5,...,1}{
\filldraw (\x,\y) circle (.05cm);
}}
\foreach \x in {-1,-.5,...,1}{
\foreach \y in {-.75,-.25,.25,.75}{
\filldraw (\x,\y) circle (.05cm);
}}
\draw[thick, blue, rounded corners=5pt] (-.375,-.375) rectangle (.375,.375);
\draw[thick, red, rounded corners=5pt] (-.875,-.875) rectangle (.875,.875);
}
\qquad
\begin{aligned}
\cL&=\text{edge lattice}
\\
r&\leq 3
\\
s&\leq 2 
\\
\Lambda &\text{ is \textcolor{blue}{blue}}
\\
\Delta &\text{ is \textcolor{red}{red}}
\end{aligned}
$$
\item
$\Lambda$ is \emph{$($incompletely$)$ $s$-surrounded} by $\Delta$, denoted $\Lambda \Subset_s \Delta$, if $\Lambda \subset_s \Delta$ and $\partial \Lambda \cap \partial \Delta$ is a non-empty 1D interval in $\cL$, which lies on exactly 1 side of $\Lambda$ and $\Delta$.
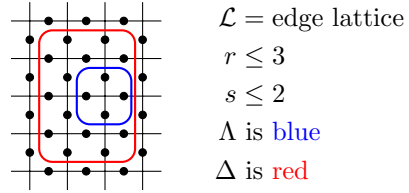
$$
\tikzmath{
\draw[scale=.5] (-2.5,-2.5) grid (1.5,2.5);
\foreach \x in {-.75,-.25,.25}{
\foreach \y in {-1,-.5,...,1}{
\filldraw (\x,\y) circle (.05cm);
}}
\foreach \x in {-1,-.5,0,.5}{
\foreach \y in {-.75,-.25,.25,.75}{
\filldraw (\x,\y) circle (.05cm);
}}
\draw[thick, blue, rounded corners=5pt] (-.375,-.375) rectangle (.35,.375);
\draw[thick, red, rounded corners=5pt] (-.875,-.875) rectangle (.4,.875);
}
\qquad
\begin{aligned}
\cL&=\text{edge lattice}
\\
r&\leq 3
\\
s&\leq 2 
\\
\Lambda &\text{ is \textcolor{blue}{blue}}
\\
\Delta &\text{ is \textcolor{red}{red}}
\end{aligned}
$$
In the cartoon above, $\partial \Lambda \cap \partial \Delta$ is the single interval consisting of one point on the right hand side.
\end{itemize}
\end{defn}

We now choose a half-plane $\bbH\subseteq \bbZ^2$ whose boundary $\partial \bbH$ intersects $\cL$ in a 1D sublattice $\partial\bbH\cong \bbZ$.
We identify $\partial\bbH=\bbZ$ below.

\begin{defn}[Boundary algebras]
\label{defn:BoundaryAlgebras}
For an interval $I\subseteq \partial\bbH= \bbZ$, 
we define:
\begin{itemize}
\item
$\Lambda_I\subseteq \bbH$ is the smallest sufficiently large rectangle with $\partial \Lambda_I \cap \partial\bbH = I$,
and
\item
$\Delta_I\subseteq \bbH$ is the smallest rectangle such that $\Lambda_I \Subset_s \Delta_I$ and $\partial \Lambda_I \cap \partial \Delta_I = I$.
\end{itemize}
The \emph{boundary algebra} $\fB(I)$ is the $\rm C^*$-algebra consisting of all operators of the form $xp_{\Delta_I}$
such that
\begin{itemize}
    \item $x \in p_{\Lambda_I} \fA(\Lambda_I) p_{\Lambda_I}$ and
    \item
    $x p_{\Delta} = p_{\Delta}x$ whenever  $\Lambda_I\Subset_s \Delta$ with $\partial \Lambda_I \cap \partial \Delta = I$.
\end{itemize}
\end{defn}

We provide local topological order axioms that are equivalent to the ones given in \cite{2307.12552}.

\begin{defn}[Local topological order]
\label{defn:LTOaxioms}
We say that our net of projections has \emph{local topological order} if the following hold:
\begin{enumerate}[label=(LTO\arabic*)]
\item 
\label{LTO:CompletelySurrounds}
If $\Lambda \ll_s \Delta$, then 
\[
p_\Delta \fA(\Lambda) p_\Delta
=
\bbC p_\Delta,
\]

\item 
\label{LTO:BoundaryAlgebras}
If $\Lambda \Subset_s \Delta\subset \bbH$ with $\partial \Lambda \cap \partial \Delta = I \subset \partial \bbH$ and $I\neq \emptyset$, then 
\[
p_\Delta \fA(\Lambda) p_\Delta
=
\fB(I) p_\Delta,
\]

\item 
\label{LTO:Injectivity}
If $\Lambda_I \Subset_s \Delta\subset\bbH$ with $\partial \Delta \cap \partial \Lambda_I = I\subset \partial\bbH$ and $I\neq \emptyset$, then for $x \in \fB(I)$, $x p_\Delta = 0$ implies $x = 0$.  
\end{enumerate}
\end{defn}

\begin{construction}
\label{const:BoundaryAlgebra}
Given a net of projections satisfying the local topological order axioms, the algebras $\fB(I)$ form a net of algebras.  
In particular, we have that if $I \subseteq J$, then we have an inclusion $\fB(I) \hookrightarrow \fB(J)$ given by $x \mapsto x p_{\Delta_I}$.
This map is injective by \ref{LTO:Injectivity}.  
Furthermore, $I \mapsto \fB(I)$ forms an inductive system, so we may take the inductive limit $\fB \coloneqq \varinjlim_I \fB(I)$.  
One can also show that if $I \cap J = \emptyset$, then $[\fB(I), \fB(J)] = 0$.
Hence $I \mapsto \fB(I)$ forms a net of algebras as defined in \cite{2304.00068, 2307.12552}.  
\end{construction}

Given a net of projections satisfying \ref{LTO:CompletelySurrounds}, we have a pure state $\psi \colon \fA \to \bbC$ given as follows: for $x \in \fA(\Lambda)$, we have that $\psi(x)$ is the scalar given by 
\[
p_\Delta x p_\Delta = \psi(x) p_\Delta,
\]
where $\Delta$ is any rectangle satisfying that $\Lambda \ll_s \Delta$.  
The fact that $\psi$ is well-defined and is a pure state is shown in \cite{2307.12552}.  
Importantly, $\psi(p_\Lambda) = 1$ for any rectangle $\Lambda$, which implies that $\psi(xp_\Lambda) = \psi(x)$ for all $x \in \fA$.
Furthermore, $\psi$ is translation invariant assuming that the net of projections is as well.  
We also have that $\psi$ extends to a state $\psi_\fB$ on $\fB$, given for $x \in \fB(I)$ by $\psi_\fB(x) \coloneqq \psi(x)$.
Note that $\psi_\fB$ is well-defined since $\psi(xp_\Lambda) = \psi(x)$ for all $x \in \fA$.

Now, consider the half plane $\bbH$ used in Definition \ref{defn:BoundaryAlgebras}, and let $\fA_{\bbH}$ be the $\mathrm{C}^*$-algebra generated by $\fA(\Lambda)$ for $\Lambda \subseteq \bbH$.
We then have a quantum channel $\bbE \colon \fA_{\bbH} \to \fB$ given for $x \in \fA(\Lambda)$ by 
\[
p_\Delta x p_\Delta = \bbE(x)p_\Delta,
\]
where $\Delta$ is any rectangle satisfying that $\Lambda \Subset_s \Delta$ with $\partial \Lambda \cap \partial \Delta \neq \emptyset$.  
The fact that $\bbE$ is a well-defined unital completely positive map is proven in \cite{2307.12552}.
Furthermore, $\bbE(x) = x$ if $x \in \fB$.

\section{Boundary algebras of the Kitaev Quantum Double model} \label{sec: Boundary Algebra}

\subsection{Quantum double model} \label{subsec:qd}

In this section we give a brief overview of the Quantum Double model originally defined in \cite{MR1951039}.  
Consider a finite group $G$. 
Given a square 2D edge lattice $\cL$, we associate $\bbC^G$-spins to each edge of the lattice.
The computational orthonormal basis is given by $\set{|g\rangle}{g\in G}$.
Observe that $\set{L_g P_h}{g,h\in G}$ forms a system of matrix units for $B(\bbC^G)$ 
where 
$L_g$ is left translation by $g$ and
$P_h$ is the minimal projection onto $|h\rangle$:
$$
L_g|k \rangle \coloneqq |gk\rangle
\qquad\qquad
P_h|k\rangle \coloneqq \delta_{h=k}|k\rangle.
$$
The operators $L_g,P_h$ satisfy the following \emph{exchange relation}:
$$
L_g P_h = | gh \rangle \langle h | = P_{gh} L_{g}.
$$
Since the action of $G$ on itself is transitive, for each $g,h\in G$, there is a unique $k\in G$ such that $L_k P_g = P_h L_k$; this $k$ is necessarily equal to $hg^{-1}$.

Instead of left translation operators, we can also use the right translation operators $R_g$ given by 
$$
R_g |k\rangle = |kg\rangle.
$$
We have the similar exchange relation
$$
R_g P_h = | hg \rangle \langle h | = P_{hg} R_{g}.
$$

\begin{nota}
When only discussing a single site with $\bbC^G$-spins, we write our operators as $P_g,L_g,R_g$.
When there are multiple sites, we indicate the location of operators by either drawing pictures with sites labelled by these operators, or we write $P^g_j, L^g_j, R^g_j$ for the operators $P_g,L_g,R_g$ respectively acting at site $j$.
\end{nota}

We now define the interaction terms for our Hamiltonian.  
The only nonzero interaction terms correspond to stars and plaquettes in the edge lattice, illustrated below: 
\[
\tikzmath{
\draw[scale=.5] (-2.5,-2.5) grid (2.5,2.5);
\foreach \x in {-.75,-.25,.25,.75}{
\foreach \y in {-1,-.5,...,1}{
\filldraw (\x,\y) circle (.05cm);
}}
\foreach  \x in {-1,-.5,...,1}{
\foreach \y in {-.75,-.25,.25,.75}{
\filldraw (\x,\y) circle (.05cm);
}}
\draw[thick, blue] (0,0) rectangle (.5,-.5);
\node[blue] at (.25,-.25) {$\scriptstyle p$};
\filldraw[blue] (0,-.25) circle (.05cm);
\filldraw[blue] (.25,0) circle (.05cm);
\filldraw[blue] (.5,-.25) circle (.05cm);
\filldraw[blue] (.25,-.5) circle (.05cm);
\draw[thick, red] (-1,.5) -- (0,.5);
\draw[thick, red] (-.5,0) -- (-.5,1);
\node[red] at (-.6,.4) {$\scriptstyle s$};
\filldraw[red] (-.75,.5) circle (.05cm);
\filldraw[red] (-.5,.25) circle (.05cm);
\filldraw[red] (-.25,.5) circle (.05cm);
\filldraw[red] (-.5,.75) circle (.05cm);
}
\]
The star term $A_s$ is the projection
$$
A_s 
\coloneqq 
\sum_{gh=\ell k}
\tikzmath{
\draw (-1,0) -- (1,0);
\draw (0,-1) -- (0,1);
\node at (-.1,-.1) {$\scriptstyle s$};
\filldraw (-.5,0) node[above]{$\scriptstyle P_g$} circle (.05cm);
\filldraw (0,-.5) node[right]{$\scriptstyle P_h$} circle (.05cm);
\filldraw (.5,0) node[above]{$\scriptstyle P_k$} circle (.05cm);
\filldraw (0,.5) node[right]{$\scriptstyle P_\ell$} circle (.05cm);
}
$$
and the plaquette term $B_p\coloneqq \frac{1}{|G|}\sum_{g\in G}B^{(g)}_p$ is the average of the translation operators
\begin{equation}
\label{eq:QDBp}
B_p^{(g)} \coloneqq
\tikzmath{
\draw (-.5, -.5) rectangle (.5, .5);
\filldraw (-.5, 0) node[left]{$\scriptstyle R_{g^{-1}}$} circle (0.05cm);
\filldraw (.5, 0) node[right]{$\scriptstyle L_{g}$} circle (0.05cm);
\filldraw (0, .5) node[above]{$\scriptstyle R_{g^{-1}}$} circle (0.05cm);
\filldraw (0, -.5) node[below]{$\scriptstyle L_g$} circle (0.05cm);
\node at (0, 0) {$\scriptstyle p$};
}
\,.
\end{equation}
We observe that $\left[A_s, B_p^{(g)}\right] = 0$ for all $s, p \subset \Lambda$ and $g \in G$.  
For a rectangle $\Lambda$, our local commuting projector Hamiltonian is 
\[
H_\Lambda
\coloneqq
\sum_{s \subset \Lambda} (I - A_s) + \sum_{p \subset \Lambda} (I - B_p),
\]
and the projection onto the local ground state space is
\[
p_{\Lambda} \coloneqq \prod_{s\subset \Lambda}A_s \prod_{p\subset \Lambda}B_p.
\]
Observe that $p_\Lambda$ absorbs all $A_s$ and $B_p^{(g)}$ for $s,p\subset \Lambda$ and $g \in G$.

\begin{defn}
Let $\Lambda \subseteq \cL$ be a rectangle, and let $c \colon \Lambda \to G$.
Note that we identify $\Lambda$ with the set of edges contained in $\Lambda$; we will continue to do this without further comment.
We define the operators 
$L_c\coloneqq\bigotimes_{\ell\in \Lambda} L^{c(\ell)}_\ell$
and 
$P_c\coloneqq\bigotimes_{\ell\in \Lambda} P^{c(\ell)}_\ell$.
Observe that every $x\in \fA(\Lambda)$ can be written as a linear combination of operators of the form $L_{c_1}P_{c_2}$.

We say $c\colon \Lambda \to G$ is \emph{flat} if $P_c=A_sP_cA_s$ for all stars $s\subset \Lambda$, and we call $P_c$ a \emph{flat operator}.
\end{defn}

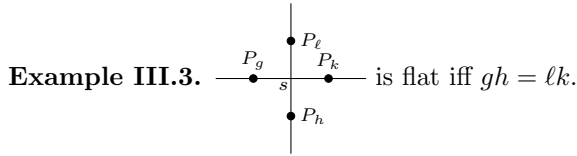
\begin{example}
$\tikzmath{
\draw (-1,0) -- (1,0);
\draw (0,-1) -- (0,1);
\node at (-.1,-.1) {$\scriptstyle s$};
\filldraw (-.5,0) node[above]{$\scriptstyle P_g$} circle (.05cm);
\filldraw (0,-.5) node[right]{$\scriptstyle P_h$} circle (.05cm);
\filldraw (.5,0) node[above]{$\scriptstyle P_k$} circle (.05cm);
\filldraw (0,.5) node[right]{$\scriptstyle P_\ell$} circle (.05cm);
}$
is flat iff $gh=\ell k$.
\end{example}

\subsection{Local topological order axioms}

We now prove that the LTO axioms in Definition \ref{defn:LTOaxioms} hold for our model, where the surrounding constant is $s=2$.  
To ease the notation, we write $\ll$ and $\Subset$ to denote $\ll_2$ and $\Subset_2$.  

The version of \ref{LTO:CompletelySurrounds} originally due to \cite{MR2742836} was proven in \cite{Cui2020kitaevsquantum}.
Another proof can be adapted from \cite[Thm.~12.1.3 and Lem.~12.1.2]{NaaijkensThesis}; we rapidly recall this latter strategy.

\begin{thm}[\cite{Cui2020kitaevsquantum,NaaijkensThesis}] 
\label{thm:LTO1QuantumDouble}
The axiom \ref{LTO:CompletelySurrounds} holds for Kitaev's Quantum Double model; i.e., if $\Lambda \ll \Delta$, then $p_\Delta \fA(\Lambda) p_\Delta = \bbC p_\Delta$. 
\end{thm}

The second proof proceeds in the following steps.
First, we may assume that $x\in \fA(\Lambda)$ is of the form $L_{c_1}P_{c_2}$ for $c_1,c_2\colon \Lambda\to G$.
\begin{enumerate}[label=\underline{Step \arabic*:}]
\item 
If $p_\Delta xp_\Delta \neq 0$, then $p_\Delta xp_\Delta=p_\Delta P_{c_2}p_\Delta$, and $c_2$ is necessarily flat.
\item
Whenever $f_1,f_2\colon\Lambda\to G$ are flat, 
$ p_\Delta P_{f_1} p_\Delta= p_\Delta P_{f_2} p_\Delta$.
\end{enumerate}
Since $\sum_{\text{flat }f}P_f = \prod_{s\subset \Lambda} A_s$, $p_{\Delta}\sum_{f\text{ flat}}P_fp_{\Delta} = p_{\Delta}$, so we can conclude from the above two steps that
$$
p_\Delta xp_\Delta
=
p_\Delta P_{c_2} p_\Delta
=
\frac{1}{\text{number of flat $c$ on $\Lambda$}}\cdot p_\Delta.
$$

Since we will modify the above steps to prove \ref{LTO:BoundaryAlgebras} below, we include complete proofs of the above steps.

\begin{proof}[Proof of Step 1]
The strategy has two parts.
First, we use the plaquette translation operators $B^{(g)}_p$ for plaquettes $p \subset \Lambda$ to cancel off the $L_{c_1}$ term, i.e.,
$B^{(g_1)}_{p_1}\cdots B^{(g_n)}_{p_n} L_{c_1}P_{c_2}=P_{c_2}$.
Since $p_\Delta$ absorbs every $B^{(g)}_p$, this proves that 
\begin{align*}
p_\Delta L_{c_1}P_{c_2}p_\Delta 
&=
p_\Delta B^{(g_1)}_{p_1}\cdots B^{(g_n)}_{p_n} L_{c_1}P_{c_2}p_\Delta
\\&=
p_\Delta P_{c_2}p_\Delta.
\end{align*}
Second, the only way $p_\Delta P_{c_2}p_\Delta\neq 0$ is if $c_2$ is flat, as $A_s P_{c_2} A_s \neq 0$ for all $s \subset \Delta$ if and only if $c_2$ is flat, and $p_\Delta = \prod_{s \subset \Delta} A_s \prod_{p \subset \Delta} B_p$.

The cancellation algorithm proceeds as follows.
Without loss of generality, we will assume that $\Lambda$ has the form shown below. 
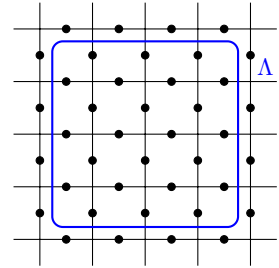
\[
\tikzmath{
\draw[scale = .7] (-2.5,-2.5) grid (2.5, 2.5);
\foreach \x in {-1.05, -.35, ..., 1.05}
{
\foreach \y in {-1.4, -.7, ..., 1.4}{
\filldraw (\x, \y) circle (0.05cm);
}
}
\foreach \y in {-1.05, -.35, ..., 1.05}
{
\foreach \x in {-1.4, -.7, ..., 1.4}{
\filldraw (\x, \y) circle (0.05cm);
}
}
\draw[thick, blue, rounded corners, scale=0.65] (-1.9, 1) -- (-1.9, -1.9) -- (1.9, -1.9) -- (1.9, 1.9) -- (-1.9, 1.9) -- (-1.9, 1);
\node[scale=.9][blue] at (1.6,.9) {$\Lambda$};
}
\]
We proceed from the left boundary of $\Lambda$. 
First, we may assume that $L^{c_1(\ell)}_\ell = I$ for all leftmost edges $\ell$, as otherwise $A_sxA_s = 0$, where $s \subset \Delta$ is the star that intersects $\Lambda$ at exactly the edge $\ell$.
We now consider the next column of smooth edges. 
If $\ell$ in this column is an outermost edge, then $L^{c_1(\ell)}_\ell = I$ for the same reason as above. 
If $\ell$ is not an outermost edge, then due to the relations
$$
R_gL_hP_k = L_{hkgk^{-1}}P_k ~\text{ and }~ L_gL_hP_k=L_{gh}P_k,
$$
we can apply a $B_p^{(g)}$ to the right of $l$ such that $B_p^{(g)}L_{c_1}P_{c_2} = L_{c_1'}P_{c_2}$, where $L_{c_1'}(\ell) = I$. We repeat this procedure for other non-outermost edges so that $L^{c_1'(\ell)}_\ell = I$ for all edges in this column. Note that $L^{c_1'(\ell)}_\ell = I$ for all edges in the next rough column, otherwise $A_sxA_s = 0$ for some $A_s$.

Moving left-to-right in this way, we eventually obtain that $B^{(g_1)}_{p_1}\cdots B^{(g_n)}_{p_n} L_{c_1}P_{c_2} = L_{c_1'}P_{c_2}$, where $L_{c_1'}(\ell) = I$ for all edges to the left of the rightmost smooth column. In particular, $L_{c_1'}$ is supported on at most two consecutive columns of edges.
We then have that $L_{c_1'} = I$ using the argument for the rough edges repeatedly, working from the outside in.
\end{proof}

\begin{proof}[Proof of Step 2] We will prove that $p_{\Delta}P_{f_1}p_{\Delta} = p_{\Delta}P_{f_2}p_{\Delta}$ for any two flat operators $P_{f_1}$ and $P_{f_2}$, and hence prove the theorem. Without loss of generality, we will assume that $\Lambda$ has the form shown below, where all four sides of $\Lambda$ have a smooth boundary.

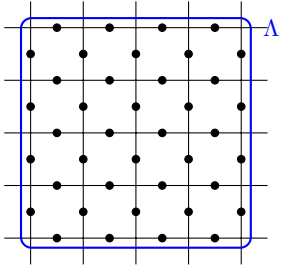
\[
\tikzmath{
\draw[scale = .7] (-2.5,-2.5) grid (2.5, 2.5);
\foreach \x in {-1.05, -.35, ..., 1.05}
{
\foreach \y in {-1.4, -.7, ..., 1.4}{
\filldraw (\x, \y) circle (0.05cm);
}
}
\foreach \y in {-1.05, -.35, ..., 1.05}
{
\foreach \x in {-1.4, -.7, ..., 1.4}{
\filldraw (\x, \y) circle (0.05cm);
}
}
\draw[thick, blue, rounded corners, scale=0.65] (-2.35, 1) -- (-2.35, -2.35) -- (2.35, -2.35) -- (2.35, 2.35) -- (-2.35, 2.35) -- (-2.35, 1);
\node[scale=.9][blue] at (1.8,1.4) {$\Lambda$};
}
\]

To prove this, we first show that there exist $g_i \in G$ and $p_i \subset \Delta$ such that 
\[
\prod_{i = 1}^m B_{p_i}^{(g_i)} P_{f_1} \prod_{i = 1}^m B_{p_i}^{(g_i^{-1})} = P_{f_1'},
\]
where $P_{f_1'}$ and $P_{f_2}$ agree on all edges except for those on any smooth boundary of $\Lambda$. 
(Note that in the particular region $\Lambda$ depicted above, we will actually have $p_i\subset \Lambda$, but this need not be the case more generally.)
Observe that $f_1' \colon \Lambda \to G$ is necessarily flat since each $B_p^{(g)}$ commutes with every $A_s$ (and hence preserves the image of each $A_s$).  
Then, by applying more $B_p^{(g)}$ terms with $p \subset \Delta$, we will have that 
\[
\prod_{i = 1}^n B_{p_i}^{(g_i)} P_{f_1} \prod_{i = 1}^n B_{p_i}^{(g_i^{-1})} = P_{f_2}.
\]
Since $p_{\Delta}$ absorbs every $B_p^{(g)}$, this shows that $p_{\Delta}P_{f_1}p_{\Delta} = p_{\Delta}P_{f_2}p_{\Delta}$.

Let $P_{f_1}$ and $P_{f_2}$ be two flat operators.  
We fix the edges not along a smooth boundary by proceeding from the rightmost column of such edges. 
Due to the relations
\[
R_{g}P_hR_{g^{-1}} = P_{hg} ~\text{ and }~ L_{g}P_hL_{g^{-1}} = P_{gh},
\]
for each inner edge of this column, going from the top-most inner edge to the bottom-most, we apply some $B_p^{(g)}$ where $p$ is the face below it to obtain 
\[
\prod_{i = 1}^{k} B_{p_i}^{(g_i)}P_{f_1}\prod_{i = 1}^{k} B_{p_i}^{(g_i^{-1})} = P_{f_1'},
\]
where $P_{f_1'}$ agrees with $P_{f_2}$ on this column of edges and is a flat operator.

We move to the next column of edges. 
Due to the same exchange relations, for each edge, we apply a $B_p^{(g)}$ where $p$ is the face to the left of it to obtain 
\[
\prod_{i = 1}^{2k+1} B_{p_i}^{(g_i)}P_{f_1}\prod_{i = 1}^{2k+1} B_{p_i}^{(g_i^{-1})} = P_{f_1'},
\]
where $P_{f_1'}$ agrees with $P_{f_2}$ on this column of edges as well. Since $P_{f_1'}$ is a flat operator, it also agrees with $P_{f_2}$ on the next column of edges. We repeat this procedure by applying $B_p^{(g)}$ operators to the columns consisting of vertically-oriented edges, and we obtain that 
\[
\prod_{i = 1}^m B_{p_i}^{(g_i)}P_{f_1}\prod_{i = 1}^m B_{p_i}^{(g_i^{-1})} = P_{f_1'},
\]
where $P_{f_1'}$ and $P_{f_2}$ agree on all but the edges along any smooth boundary of $\Lambda$.

We now deal with the edges along the smooth boundaries of $\Lambda$.
For any such edge $\ell$, we apply some $B_p^{(g)}$ where the plaquette $p \subset \Delta$ intersects $\Lambda$ at exactly this edge.  
We then obtain $\prod_{i = 1}^n B_{p_i}^{(g_i)}P_{f_1}\prod_{i = 1}^n B_{p_i}^{(g_i^{-1})} = P_{f_2}$.
\end{proof}

We now prove \ref{LTO:BoundaryAlgebras}.  
In our proof, we will consider the case $\Lambda \Subset \Delta$ where $I \coloneqq \partial \Lambda \cap \partial \Delta$ is nonempty and vertical and $\Lambda$, $\Delta$ both lie to the left of $I$.  
Using the notation in Definition \ref{defn:BoundaryAlgebras}, we have that $\partial\bbH$ is vertical and $\bbH$ is the left half plane.  
The other cases can be handled similarly. 

Our proof of \ref{LTO:BoundaryAlgebras} will make use of the algebra $\fC(I)$, which we later show is isomorphic to $\fB(I)$.  
Let $\widetilde I \subseteq \Lambda$ be the set comprising the edges in $I$ and the column adjacent to $I$.
The algebra $\fC(I)$ if $I$ is a rough interval is given as below:
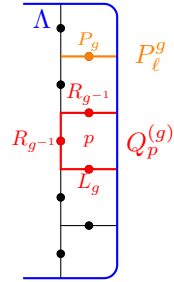
\[
\tikzmath{
\draw (1.5,.05) -- (1.5,3.7);
\foreach \y in {.75,1.5,2.25,3}{
    \draw (2.25,\y) -- (1.5,\y);
}
\draw[thick, red] (2.25,1.5) -- (1.5,1.5) -- (1.5,2.25) -- (2.25,2.25);
\foreach \x in {1.875}{
\foreach \y in {.75, 1.5, 2.25, 3}{
\filldraw (\x, \y) circle (0.05cm);
}}
\foreach \x in {1.5}{
\foreach \y in {.375, 1.125, 1.875, 2.625, 3.375}{
\filldraw (\x, \y) circle (0.05cm);
}}
\foreach \x in {1.875}{
\foreach \y in {1.5, 2.25}{
\filldraw[red] (\x,\y) circle (.05cm);
}}
\foreach \x in {1.5}{
\foreach \y in {1.875}{
\filldraw[red] (\x,\y) circle (.05cm);
}}
\foreach \x in {1.875}{
\foreach \y in {3}{
\filldraw[orange] (\x,\y) circle (.05cm);
}}
\node[red][scale = 0.8] at (1.15,1.875) {$R_{g^{-1}}$};
\node[red][scale = 0.8] at (1.875,1.3) {$L_g$};
\node[red][scale = 0.8] at (1.875,2.5) {$R_{g^{-1}}$};
\node[red] at (2.7,1.875) {$Q_p^{(g)}$};
\node[red] at (1.875,1.875) {$\scriptstyle p$};
\draw[thick, orange] (1.5,3) -- (2.25,3);
\node[orange] at (1.875,3.2) {$\scriptstyle P_g$};
\node[orange] at (2.7,3) {$P_\ell^{g}$};
\node at (2,-.5) {$\mathfrak{C}(I) \coloneqq
{ C^*}\set{P_\ell^{g_1},Q_p^{(g_2)}}{ \ell\subset I,p \subset \widetilde{I},g_1,g_2\in G}$};
\draw[thick, blue, rounded corners=5pt] (1,3.7) -- (2.25,3.7) -- (2.25,.05) -- (1,.05);
\node[blue] at (1.25,3.5) {$\Lambda$}; 
}\]

Next, we define $S_s^{(g)}$ to be the projection
$$
S_s^{(g)} 
\coloneqq 
\sum_{k\ell=hg}
\tikzmath{
\draw (-1,0) -- (0,0);
\draw[dotted] (0,0) -- (1,0);
\draw (0,-1) -- (0,1);
\node at (-.1,-.1) {$\scriptstyle s$};
\filldraw (-.5,0) node[above]{$\scriptstyle P_k$} circle (.05cm);
\filldraw (0,-.5) node[right]{$\scriptstyle P_\ell$} circle (.05cm);
\filldraw[fill=white] (.5,0) node[above]{$\scriptstyle P_g$} circle (.05cm);
\filldraw (0,.5) node[right]{$\scriptstyle P_h$} circle (.05cm);
}
$$
The dotted edge above represents a \emph{ghost edge}, which is not actually a part of $S_s^{(g)}$.
While this ghost edge would be a part of the corresponding $A_s$ operator, we delete it here to obtain a projection that only acts on 3 sites. 
The group element $g$ that would have labelled the ghost edge can be recovered from $h,k,\ell$ and the defining relation of the $A_s$ operator.
 
The algebra $\fC(I)$ for the smooth interval case is then given as below:
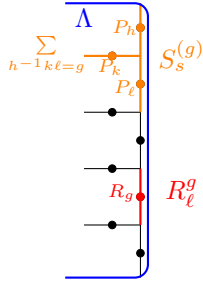
\[
\tikzmath{
\draw (2.25,.05) -- (2.25,3.7);
\foreach \y in {.75,1.5,2.25,3}{
    \draw (2.25,\y) -- (1.5,\y);
}
\foreach \x in {1.875}{
\foreach \y in {.75, 1.5, 2.25, 3}{
\filldraw (\x, \y) circle (0.05cm);
}}
\foreach \x in {2.25}{
\foreach \y in {.375, 1.125, 1.875, 2.625, 3.375}{
\filldraw (\x, \y) circle (0.05cm);
}}
\foreach  \x in {2.25}{
\foreach \y in {1.125}{
\filldraw[red] (\x,\y) circle (.05cm);
}}
\foreach  \x in {2.25}{
\foreach \y in {2.625, 3.375}{
\filldraw[orange] (\x,\y) circle (.05cm);
}}
\foreach  \x in {1.875}{
\foreach \y in {3}{
\filldraw[orange] (\x,\y) circle (.05cm);
}}
\draw[thick, red] (2.25,.75) -- (2.25,1.5);
\node[orange] at (1,3) {$\scriptstyle \sum\limits_{h^{-1}k\ell = g}$};
\node[orange] at (2.05,2.575) {$\scriptstyle P_{\ell}$};
\node[orange] at (2.05,3.375) {$\scriptstyle P_{h}$};
\node[orange] at (1.85,2.85) {$\scriptstyle P_{k}$};
\node[red] at (2,1.175) {$\scriptstyle R_g$};
\node[red] at (2.8,1.175) {$R_\ell^{g}$};
\draw[thick, orange] (1.5,3) -- (2.25,3);
\draw[thick, orange] (2.25,2.25) -- (2.25,3.7);
\node[orange] at (2.8,3) {$S_s^{(g)}$};
\node at (2,-.4) {$\mathfrak{C}(I) \coloneqq
{ C^*}\set{S_s^{(g_1)},R_\ell^{g_2}}{ \ell\subset I, s \subset \widetilde{I},g_1,g_2\in G}$};
\draw[thick, blue, rounded corners=5pt] (1.25,3.7) -- (2.35,3.7) -- (2.35,.05) -- (1.25,.05);
\node[blue] at (1.5,3.5) {$\Lambda$}; 
}
\]
The generators can be viewed as $A_s$ and $B_p^{(g)}$ operators restricted to the boundary.

\begin{thm} \label{thm:LTO2QuantumDouble} The axiom \ref{LTO:BoundaryAlgebras} holds for Kitaev's Quantum Double model; i.e., if $\Lambda \Subset \Delta$ with $\partial \Lambda \cap \partial \Delta = I \neq \emptyset$, then 
$p_\Delta \fA(\Lambda) p_\Delta = \fB(I) p_\Delta$.
\end{thm}

This proof proceeds in the following steps.  
First, we will show that $p_\Delta \fA(\Lambda) p_\Delta = \fC(I) p_\Delta$.  
This will prove the desired claim, since by how $\fB(I)$ is defined, we will have that
\[
\fC(I) p_{\Delta_I}
\subseteq
\fB(I)
\subseteq
p_{\Delta_I} \fA(\Lambda_I) p_{\Delta_I}
=
\fC(I)p_{p_{\Delta_I}}.
\]

We will prove the case where $I$ is a rough interval since the proof of the smooth interval case is analogous.
Note that as before, we may assume that $x\in \fA(\Lambda)$ is of the form $L_{c_1}P_{c_2}$ for $c_1,c_2\colon \Lambda\rightarrow G$.
We then have the following two steps, analogous to the two steps in the proof of Theorem \ref{thm:LTO1QuantumDouble}.

\begin{enumerate}[label=\underline{Step \arabic*:}]
\item If $p_{\Delta}xp_{\Delta}\neq 0$, then 
\[
p_{\Delta}xp_{\Delta} = \prod_{i = 1}^n Q_{p_i}^{(g_i)}p_{\Delta} P_{c_2}p_{\Delta}
\]
for some $Q_{p_i}^{(g_i)}$, and $c_2$ is necessarily flat. 
\item Let $S$ be the set of $c\colon\Lambda\rightarrow G$ such that $c$ is flat and $c(\ell) = c_2(\ell)$ for all $\ell\in I$. Then for all $c_3\in S$, we have that $p_{\Delta}P_{c_3}p_{\Delta} = p_{\Delta}P_{c_2}p_{\Delta}$.
\end{enumerate}

Note that since 
\[
\sum_{c\in S}P_c = \prod_{s\subset \Lambda}A_s\prod_{\ell \in I} P_{\ell}^{c(\ell)}
\]
we can conclude from the above two steps that
$$
p_\Delta xp_\Delta
=
\frac{1}{|S|}\prod_{i = 1}^n Q_{p_i}^{(g_i)}\prod_{\ell \in I} P_{\ell}^{c(\ell)}p_{\Delta}.
$$

\begin{proof}[Proof of Step 1] 
Observe that the region $\widetilde{I}$ has the following form:
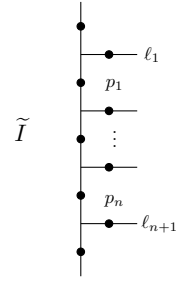
\[
\tikzmath{
\draw (1.5,.05) -- (1.5,3.7);
\foreach \y in {.75,1.5,2.25,3}{
    \draw (2.25,\y) -- (1.5,\y);
}
\foreach \x in {1.875}{
\foreach \y in {.75, 1.5, 2.25, 3}{
\filldraw (\x, \y) circle (0.05cm);
}}
\foreach \x in {1.5}{
\foreach \y in {.375, 1.125, 1.875, 2.625, 3.375}{
\filldraw (\x, \y) circle (0.05cm);
}}
\node[scale = .7] at (2.45,3) {$\ell_1$}; 
\node[scale = .7] at (2.55,.75) {$\ell_{n+1}$}; 
\node[scale = .7] at (1.95,2.6) {$p_1$}; 
\node[scale = .7] at (1.95,1.95) {$\vdots$}; 
\node[scale = .7] at (1.95,1.05) {$p_n$}; 
\node[scale = .9] at (.7,2) {$\widetilde{I}$};
}
\]
By the first step of Theorem \ref{thm:LTO1QuantumDouble}, there exist operators $B_p^{(g)}$ with $p \subset \Lambda$ such that $\prod B_p^{(g)}  L_{c_1}P_{c_2} = L_{c_1'}P_{c_2}$, where $L_{c_1'}$ is supported on $\widetilde{I}$.  
In addition, we have that $c_1'(\ell) = e$ for outermost edges $\ell$ in $\widetilde{I}\setminus I$.

Due to the exchange relations shown before, there exist operators $Q_{p_1}^{(g_1)},\cdots, Q_{p_n}^{(g_n)}$ such that $\prod_{i = 1}^nQ_{p_i}^{(g_i)}L_{c_1'}P_{c_2} = L_{c_1''}P_{c_2}$, where the support of $L_{c_1''}$ is on $I$. Furthermore, its support is empty, otherwise $A_sxA_s = 0$ for some $s\subset \Delta$ that intersects $I$. Therefore, $L_{c_1'} = \prod_{i = 1}^n Q_{p_i}^{(g_i^{-1})}$, and hence 
\[
p_{\Delta}xp_{\Delta} = \prod_{i = 1}^n Q_{p_i}^{(g_i^{-1})}p_{\Delta}P_{c_2}p_{\Delta}.
\]
In addition, $P_{c_2}$ is a flat operator, as otherwise $p_{\Delta}P_{c_2}p_{\Delta} = 0$.
\end{proof}

\begin{proof}[Proof of Step 2] Let $c_3\in S$. 
The result follows by the proof of the second step of Theorem \ref{thm:LTO1QuantumDouble}.
Specifically, we proceed from the right boundary, and we apply $B_p^{(g)}$ operators where each $p \subset \Delta$ lies to the left of a vertically-oriented edge.  
Using the argument from that proof, we have that \[
\prod_{i = 1}^k B_{p_i}^{(g_i)}P_{c_3}\prod_{i = 1}^k B_{p_i}^{(g_i^{-1})} = P_{c_2},
\]
so $p_{\Delta}P_{c_3}p_{\Delta} = p_{\Delta}P_{c_2}p_{\Delta}$. 
\end{proof}

\begin{rem}
\label{rem:DimensionsOfBoundaryAlgebras}
We observe that we have canonical bases for $\fC(I)$ in both the smooth and rough cases.  
Let $n \coloneqq |I|$.
If $I$ is rough, then the following is a basis for $\fC(I)$:
\begin{equation}
\label{eq:RoughCanonicalBasis}
\set{\prod_{i = 1}^{n - 1}Q^{(g_i)}_{p_i} \prod_{i = 1}^n P^{h_i}_{\ell_i}}{g_i, h_i \in G}.
\end{equation}
Here $p_1, \dots, p_{n - 1}$ are the $n - 1$ partial plaquettes in $\widetilde I$ and $\ell_1, \dots, \ell_n$ are the $n$ edges of $I$.  
Similarly, if $I$ is smooth, then the following is a basis for $\fC(I)$:
\begin{equation}
\label{eq:SmoothCanonicalBasis}
\set{\prod_{i = 1}^{n}R^{g_i}_{\ell_i} \prod_{i = 1}^{n - 1}S^{(h_i)}_{s_i}}{g_i, h_i \in G}.
\end{equation}
As before, $\ell_1, \dots, \ell_n$ are the $n$ edges of $I$, and $s_1, \dots, s_{n - 1}$ are the $n - 1$ partial stars in $\widetilde I$.  
Note that in both cases, the dimension of $\fC(I)$ is $|G|^{2n - 1}$.
\end{rem}

\begin{rem}
\label{rem:BasisOfImageOfPDelta}
We also provide a description of the image of $p_\Delta$ for certain rectangles $\Delta$.  
First, suppose $\Delta$ has the following form: 
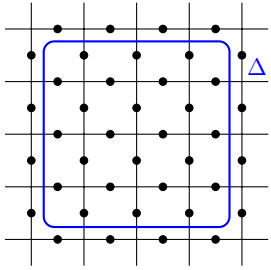
\[
\tikzmath{
\draw[scale = .7] (-2.5,-2.5) grid (2.5, 2.5);
\foreach \x in {-1.05, -.35, ..., 1.05}
{
\foreach \y in {-1.4, -.7, ..., 1.4}{
\filldraw (\x, \y) circle (0.05cm);
}
}
\foreach \y in {-1.05, -.35, ..., 1.05}
{
\foreach \x in {-1.4, -.7, ..., 1.4}{
\filldraw (\x, \y) circle (0.05cm);
}
}
\draw[thick, blue, rounded corners, scale=0.65] (-1.9, 1) -- (-1.9, -1.9) -- (1.9, -1.9) -- (1.9, 1.9) -- (-1.9, 1.9) -- (-1.9, 1);
\node[scale=.9][blue] at (1.6,.9) {$\Delta$};
}
\]
In that case, the image of $p_\Delta$ has a basis whose elements are the sum of all flat simple tensors with the same labels for the boundary edges.  
There is also a condition on the boundary labels in order for there to exist flat tensors with that particular labeling; this condition is analogous to the condition for $A_s$.  
Now, suppose $\Delta$ has the following form: 
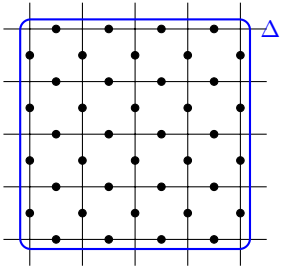
\[
\tikzmath{
\draw[scale = .7] (-2.5,-2.5) grid (2.5, 2.5);
\foreach \x in {-1.05, -.35, ..., 1.05}
{
\foreach \y in {-1.4, -.7, ..., 1.4}{
\filldraw (\x, \y) circle (0.05cm);
}
}
\foreach \y in {-1.05, -.35, ..., 1.05}
{
\foreach \x in {-1.4, -.7, ..., 1.4}{
\filldraw (\x, \y) circle (0.05cm);
}
}
\draw[thick, blue, rounded corners, scale=0.65] (-2.35, 1) -- (-2.35, -2.35) -- (2.35, -2.35) -- (2.35, 2.35) -- (-2.35, 2.35) -- (-2.35, 1);
\node[scale=.9][blue] at (1.8,1.4) {$\Delta$};
}
\]
In that case, the image of $p_\Delta$ has a basis whose elements are the sum of all flat simple tensors obtained by applying a sequence of $B_p^{(g)}$ operators to a simple tensor with every interior edge labeled by the identity.  
One can verify these characterizations by using step 2 of the proof of Theorem \ref{thm:LTO1QuantumDouble}.
\end{rem}

\begin{thm} 
\label{thm:LTO3QuantumDouble}
Suppose that we have rectangles $\Lambda \Subset\Delta$ with $\partial \Lambda \cap \partial \Delta = I \neq \emptyset$.
\begin{enumerate}[label=(\arabic*)]
\item 
If $I$ is rough, then $x\in \fC(I)$ and $xp_{\Delta} = 0$ implies that $x = 0$.
\item 
If $I$ is smooth, then $x\in \fC(I)$ and $xp_{\Delta} = 0$ implies that $x = 0$.
\end{enumerate}
\end{thm}

\begin{proof} 
We prove the second case, and the first is similar and left to the reader. 
We let $n \coloneqq |I|$. 
Without loss of generality, we may assume that $\Delta$ has the following form: 
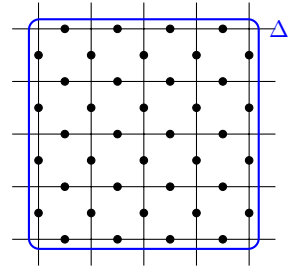
\[
\tikzmath{
\draw[scale = .7] (-2.5,-2.5) grid (2.5, 2.5);
\foreach \x in {-1.05, -.35, ..., 1.05}
{
\foreach \y in {-1.4, -.7, ..., 1.4}{
\filldraw (\x, \y) circle (0.05cm);
}
}
\foreach \y in {-1.05, -.35, ..., 1.05}
{
\foreach \x in {-1.4, -.7, ..., 1.4}{
\filldraw (\x, \y) circle (0.05cm);
}
}
\draw[thick, blue, rounded corners, scale=0.65] (-2.35, 1) -- (-2.35, -2.35) -- (2.35, -2.35) -- (2.35, 2.35) -- (-2.35, 2.35) -- (-2.35, 1);
\node[scale=.9][blue] at (1.8,1.4) {$\Delta$};
}
\]
We first prove the following lemma.  

\begin{lem}
\label{lem:LinearlyIndependentSetForLTO3Proof}
Let $|c \rangle$ be a flat simple tensor whose interior edges are all labeled by the identity.
Then the set 
\[
\set{p_\Delta \prod_{j = 1}^n R^{g_j}_{\ell_j} |c \rangle}{g_1, \dots, g_n \in G}
\] 
is linearly independent.
\end{lem}

\begin{proof}[Proof of Lemma]
For each distinct $n$-tuple of group elements $(g_1, \dots, g_n)$, we obtain a distinct flat simple tensor $\prod_{j = 1}^n R^{g_j}_{\ell_j} |c \rangle$, and each of these flat simple tensors have interior edges labeled by the identity.  
For a flat simple tensor $|v\rangle$, $p_\Delta |v\rangle$ is a scalar multiple of the sum of all flat simple tensors obtained from $|v \rangle$ by applying a sequence of $B_p^{(g)}$ terms.  
In particular, we obtain that for each distinct $n$-tuple $(g_1, \dots, g_n) \in G^n$, $p_\Delta \prod_{j = 1}^n R^{g_j}_{\ell_j} |c \rangle$ corresponds to a distinct element of the basis for the image of $p_\Delta$ described in Remark \ref{rem:BasisOfImageOfPDelta} above. 
\end{proof}

Now, let $x\in \fC(I)$.  
Then $x = a_1x_1+ \cdots + a_kx_k$, where the $x_i$ enumerate the canonical basis for $\fC(I)$ (so $k = |G|^{2n - 1}$).
Suppose $xp_{\Delta} = 0$.  
Then 
$$
xp_{\Delta} = p_{\Delta}x = \sum_{i = 1}^{k}a_ip_{\Delta}x_i= 0.$$ 
Now, let $|c \rangle$ be a flat simple tensor whose interior edges are all labeled by the identity.  
Then $|c \rangle \in \im\left( \prod_{i = 1}^{n - 1} S_{s_i}^{(h_i)}\right)$ for exactly one $(n - 1)$-tuple $(h_1, \dots, h_{n- 1}) \in G^{n - 1}$.  
Without loss of generality, we may assume $x_1, \dots, x_m$ are the basis elements satisfying that $x_i = \prod_{j = 1}^n R^{g_{i, j}}_{\ell_j} \prod_{j =1}^{n - 1} S^{(h_j)}_{s_j}$, 
where $(h_1, \dots, h_{n - 1})$ is the specific $(n - 1)$-tuple mentioned above.
(Note that $m = |G|^n$.)
Then 
\begin{align*}
0
&=
\sum_{i = 1}^k a_i p_\Delta x_i|c\rangle
=
\sum_{i = 1}^m a_i p_\Delta x_i|c\rangle
\\&=
\sum_{i=1}^m a_{i} p_\Delta \prod_{j = 1}^n R^{g_{i,j}}_{\ell_j} |c \rangle.
\end{align*}
Now, for each $i = 1, \dots, m$, $p_\Delta \prod_{j = 1}^n R^{g_{i,j}}_{\ell_j} |c \rangle$ is a distinct element of the linearly independent set described in Lemma \ref{lem:LinearlyIndependentSetForLTO3Proof}.
Hence $a_i = 0$ for all $i = 1, \dots, m$.
Letting $|c \rangle$ vary over all flat simple tensors with interior edges labeled by the identity, we obtain that every $a_i = 0$, as desired.
\end{proof}

\begin{rem}
Note that since $\fC(I)p_{\Delta_I} = \fB(I)$, by Theorem \ref{thm:LTO3QuantumDouble} we have that $\fC(I) \cong \fB(I)$.  
Hence Theorem \ref{thm:LTO3QuantumDouble} implies that \ref{LTO:Injectivity} is satisfied for Kitaev's Quantum Double model.
\end{rem}

\subsection{Abstract characterization of the boundary algebras}

In this section, we prove Theorem \ref{thm:BoundaryAlgebrasForKQD}:
\begin{enumerate}[label=(\arabic*)]
\item 
If $I$ is a rough interval with $n$ horizontal boundary edges, then we have that $\fC(I) \cong \End_{\Hilb(G)}(\bbC[G]^{\otimes n})$.
\item 
If $I$ is a smooth interval with $n$ vertical boundary edges, then $\fC(I) \cong \End_{\Rep(G)} ((\bbC^G)^{\otimes n})$.
\end{enumerate}
We provide two proofs of these results.  
The first is a direct proof, whereas the second is a quicker proof using fusion category techniques.  
Although the local algebras $\fC(I)$ are not isomorphic when $G$ is non-abelian, the limit AF $\rmC^*$-algebras $\varinjlim_I \fC(I)$ are both equivalent to the same UHF algebra $M_{|G|^\infty}$.

\subsubsection{Direct proof}

We first analyze the rough boundary algebra.
Observe that the algebra 
$$
\mathfrak{C}(I) 
=
{\rmC^*}\set{P_\ell^{g_1},Q_p^{(g_2)}}{\ell\subset I,p \subset \widetilde{I},g_1,g_2\in G}
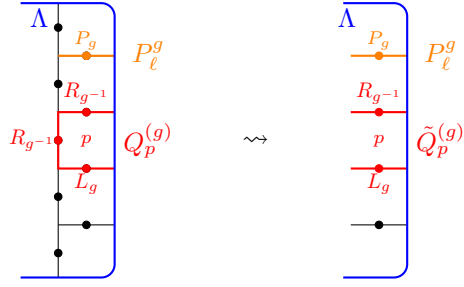
$$
\[
\tikzmath{
\draw (1.5,.05) -- (1.5,3.7);
\foreach \y in {.75,1.5,2.25,3}{
    \draw (2.25,\y) -- (1.5,\y);
}
\draw[thick, red] (2.25,1.5) -- (1.5,1.5) -- (1.5,2.25) -- (2.25,2.25);
\foreach \x in {1.875}{
\foreach \y in {.75, 1.5, 2.25, 3}{
\filldraw (\x, \y) circle (0.05cm);
}}
\foreach \x in {1.5}{
\foreach \y in {.375, 1.125, 1.875, 2.625, 3.375}{
\filldraw (\x, \y) circle (0.05cm);
}}
\foreach \x in {1.875}{
\foreach \y in {1.5, 2.25}{
\filldraw[red] (\x,\y) circle (.05cm);
}}
\foreach \x in {1.5}{
\foreach \y in {1.875}{
\filldraw[red] (\x,\y) circle (.05cm);
}}
\foreach \x in {1.875}{
\foreach \y in {3}{
\filldraw[orange] (\x,\y) circle (.05cm);
}}
\node[red][scale = 0.8] at (1.15,1.875) {$R_{g^{-1}}$};
\node[red][scale = 0.8] at (1.875,1.3) {$L_g$};
\node[red][scale = 0.8] at (1.875,2.5) {$R_{g^{-1}}$};
\node[red] at (2.7,1.875) {$Q_p^{(g)}$};
\node[red] at (1.875,1.875) {$\scriptstyle p$};
\draw[thick, orange] (1.5,3) -- (2.25,3);
\node[orange] at (1.875,3.2) {$\scriptstyle P_g$};
\node[orange] at (2.7,3) {$P_\ell^{g}$};
\draw[thick, blue, rounded corners=5pt] (1,3.7) -- (2.25,3.7) -- (2.25,.05) -- (1,.05);
\node[blue] at (1.25,3.5) {$\Lambda$}; 
}
\qquad
\rightsquigarrow
\qquad
\tikzmath{
\draw (2.25,.75) -- (1.5,.75);
\draw[thick, red] (2.25,1.5) -- (1.5,1.5);
\draw[thick, red] (1.5,2.25) -- (2.25,2.25);
\filldraw (1.875, .75) circle (0.05cm);
\foreach \x in {1.875}{
\foreach \y in {1.5, 2.25}{
\filldraw[red] (\x,\y) circle (.05cm);
}}
\filldraw[orange] (1.875,3) circle (.05cm);
\node[red][scale = 0.8] at (1.875,1.3) {$L_g$};
\node[red][scale = 0.8] at (1.875,2.5) {$R_{g^{-1}}$};
\node[red] at (2.7,1.875) {$\tilde{Q}_p^{(g)}$};
\node[red] at (1.875,1.875) {$\scriptstyle p$};
\draw[thick, orange] (1.5,3) -- (2.25,3);
\node[orange] at (1.875,3.2) {$\scriptstyle P_g$};
\node[orange] at (2.7,3) {$P_\ell^{g}$};
\draw[thick, blue, rounded corners=5pt] (1.4,3.7) -- (2.25,3.7) -- (2.25,.05) -- (1.4,.05);
\node[blue] at (1.45,3.5) {$\Lambda$}; 
}
\]
is $*$-isomorphic to the algebra
$$
\fD(I)
\coloneqq
{\rmC^*}\set{P_\ell^{g_1},\tilde{Q}_p^{(g_2)}}{ \ell,p \subset I,g_1,g_2\in G}.
$$
Thus the boundary algebra $\fC(I)$ is $*$-isomorphic to an algebra which acts only on the horizontal edges of the boundary.

We now identify the Hilbert spaces on the rough horizontal boundary edges with $\bbC[G]$ with canonical ONB $\{|g\rangle\}_{g\in G}$ viewed as a $G$-graded Hilbert space.
That is, when $I$ consists of $n$ horizontal boundary edges,
\begin{equation}
\label{eq:G-GradedONB}
\bigotimes^n_{j=1} \bbC[G]
=
\bigoplus_{g\in G} 
\underbrace{\operatorname{span}\set{\bigotimes_{j = 1}^{n} | g_{j}\rangle}{\prod_{j = 1}^{n} g_{j} = g}}_{V_g\coloneqq}
\end{equation}
is the decomposition of the tensor product Hilbert space into its $g$-graded components.
Observe that the set in \eqref{eq:G-GradedONB} above gives a distinguished ONB of $V_g$, and thus $\dim(V_g)=|G|^{n-1}$ for each $g\in G$.

\begin{thm} \label{thm:2d-qd-matrix-decomp}
Each $g$-graded subspace $V_g$ is $\fD(I)$-invariant, and $\fD(I)$ acts irreducibly on each $V_g$ as a full matrix algebra isomorphic to $M_{|G|^{n - 1}} (\mathbb{C})$.
Thus
$$
\fD(I) \cong \End_{\Hilb(G)}(\bbC[G]^{\otimes n}).
$$
\end{thm}
\begin{proof}
First, observe that the generators of $\fD(I)$ map elements of the distinguished ONB of $V_g$ from \eqref{eq:G-GradedONB} to other elements of the ONB, and thus $V_g$ is $\fD(I)$-invariant.

It remains to prove $\fD(I)$ acts irreducibly on each $V_g$.
To do so, we construct a system of matrix units in the restriction of $\fD(I)$ to $V_g$.
We denote the restriction of the generators to $V_g$ by the same name to ease the notation.

When $g_1\cdots g_n=g$, observe $P^{g_1}_1\cdots P^{g_n}_n$ is the minimal projection onto $\bbC\bigotimes_{j = 1}^{n} | g_{j}\rangle\subset V_g$, where the subscript denotes that $P^{g_i}_i$ acts at site $i$.
When in addition $h_1\cdots h_n=g$,
the partial isometry from 
$\bbC\bigotimes_{j = 1}^{n} | g_{j}\rangle$
onto
$\bbC\bigotimes_{j = 1}^{n} | h_{j}\rangle$
is given by
$$
P^{h_1}_1\cdots P^{h_n}_n
x
P^{g_1}_1\cdots P^{g_n}_n
$$
where 
$x$ is a product of generators of type $\tilde Q^{(k)}_p$:
\begin{align*}
x=&(L^{h_{n-1}^{-1}\cdots h_2^{-1}h_1^{-1}g_1g_2\cdots g_{n-1}}_nR_{n-1}^{g_{n-1}^{-1}\cdots g_2^{-1}g_1^{-1}h_1h_2\cdots h_{n-1}})
\\&
\cdots
(L^{h_2^{-1}h_1^{-1}g_1g_2}_3R^{g_2^{-1}g_1^{-1}h_1h_2}_2)(L^{h_1^{-1}g_1}_2 R^{g_1^{-1}h_1}_1).
\end{align*}
These $n-1$ operators switch the first $n-1$ $|g_i\rangle$ to $|h_i\rangle$ and the last $|g_n\rangle$ to 
$$
|\underbrace{h_{n-1}^{-1}\cdots h_2^{-1}h_1^{-1}}_{=h_ng^{-1}}\underbrace{g_1g_2\cdots g_{n-1}g_{n}}_{=g}\rangle
=
|h_n\rangle
$$
as promised.
\end{proof}

We now turn our attention to the smooth boundary algebra.
As with the rough case, we first make the observation that $$
\fC(I)=
{\rmC^*}\set{S_s^{(g_1)},R_\ell^{g_2}}{\ell\subset I, s \subset \widetilde{I},g_1,g_2\in G}
$$
\begin{equation}   
\label{eq:RemoveUnneccessaryEdges}
\tikzmath{
\draw (2.25,.05) -- (2.25,3.7);
\foreach \y in {.75,1.5,2.25,3}{
    \draw (2.25,\y) -- (1.5,\y);
}
\foreach \x in {1.875}{
\foreach \y in {.75, 1.5, 2.25, 3}{
\filldraw (\x, \y) circle (0.05cm);
}}
\foreach \x in {2.25}{
\foreach \y in {.375, 1.125, 1.875, 2.625, 3.375}{
\filldraw (\x, \y) circle (0.05cm);
}}
\foreach  \x in {2.25}{
\foreach \y in {1.125}{
\filldraw[red] (\x,\y) circle (.05cm);
}}
\foreach  \x in {2.25}{
\foreach \y in {2.625, 3.375}{
\filldraw[orange] (\x,\y) circle (.05cm);
}}
\foreach  \x in {1.875}{
\foreach \y in {3}{
\filldraw[orange] (\x,\y) circle (.05cm);
}}
\draw[thick, red] (2.25,.75) -- (2.25,1.5);
\node[orange] at (1,3) {$\scriptstyle \sum\limits_{h^{-1}k\ell = g}$};
\node[orange] at (2.05,2.575) {$\scriptstyle P_{\ell}$};
\node[orange] at (2.05,3.375) {$\scriptstyle P_{h}$};
\node[orange] at (1.85,2.85) {$\scriptstyle P_{k}$};
\node[red] at (2,1.175) {$\scriptstyle R_g$};
\node[red] at (2.8,1.175) {$R_\ell^{g}$};
\draw[thick, orange] (1.5,3) -- (2.25,3);
\draw[thick, orange] (2.25,2.25) -- (2.25,3.7);
\node[orange] at (2.8,3) {$S_s^{(g)}$};
\draw[thick, blue, rounded corners=5pt] (1.25,3.7) -- (2.35,3.7) -- (2.35,.05) -- (1.25,.05);
\node[blue] at (1.5,3.5) {$\Lambda$}; 
}
\qquad\rightsquigarrow\qquad
\tikzmath{
\draw (2.25,.05) -- (2.25,3.7);
\foreach \x in {2.25}{
\foreach \y in {.375, 1.125, 1.875, 2.625, 3.375}{
\filldraw (\x, \y) circle (0.05cm);
}}
\filldraw[red] (2.25,1.125) circle (.05cm);
\foreach  \x in {2.25}{
\foreach \y in {2.625, 3.375}{
\filldraw[orange] (\x,\y) circle (.05cm);
}}
\draw[thick, red] (2.25,.75) -- (2.25,1.5);
\node[orange] at (1.75,3) {$\scriptstyle \sum\limits_{h^{-1}\ell = g}$};
\node[orange] at (2.05,2.575) {$\scriptstyle P_{\ell}$};
\node[orange] at (2.05,3.375) {$\scriptstyle P_{h}$};
\node[red] at (2,1.175) {$\scriptstyle R_g$};
\node[red] at (2.8,1.175) {$R_\ell^{g}$};
\draw[thick, orange] (2.25,2.25) -- (2.25,3.7);
\node[orange] at (2.8,3) {$\widetilde{S}_s^{(g)}$};
\draw[thick, blue, rounded corners=5pt] (1.5,3.7) -- (2.35,3.7) -- (2.35,.05) -- (1.5,.05);
\node[blue] at (1.75,3.5) {$\Lambda$}; 
}
\end{equation}
is $*$-isomorphic to
$$
\fD(I)\coloneqq
{\rmC^*}\set{\tilde{S}_s^{(g_1)},R_\ell^{g_2}}{ s,\ell \subset I,g_1,g_2\in G}.
$$
Thus similar to before, the boundary algebra $\fC(I)$ is $*$-isomorphic to an algebra which acts only on the vertical edges of the boundary.

We now identify the Hilbert spaces on the smooth vertical boundary with $\bbC^G=\Fun(G\to \bbC)$, which each carry the left regular representation given by $(L_gf)(h)=f(g^{-1}h)$.
In bra-ket notation,
$|g\rangle\colon G\to \bbC$ is the Dirac delta function at $g$,
and $L_g|h\rangle = |gh\rangle$.
Observe there is a diagonal $G$-action on the tensor product Hilbert space $\bigotimes_{i=1}^n \bbC^G$
given by $U_g\bigotimes_{j=1}^n f_j \coloneqq \bigotimes_{j=1}^n L_g f_j$, so we may view the tensor product Hilbert space as an object in $\Rep(G)$.

\begin{thm}
The $\rmC^*$-algebra $\fD(I)$ acting on $\bigotimes_{j=1}^n \bbC^G$ is exactly the $G$-equivariant endomorphisms, i.e., 
$$
\fD(I)=\set{T\in \End( \bbC^G)^{\otimes n}}{TU_g=U_gT \quad\forall g\in G}.
$$
In particular,
$\fD(I)=\End_{\Rep(G)}(\bigotimes_{j=1}^n \bbC^G)$.
\end{thm}

\begin{proof}
Observe that the generators of $\fD(I)$ commute with the diagonal $G$-action on $(\bbC^G)^{\otimes n}$, so $\fD(I)\subseteq \End_{\Rep(G)}((\bbC^G)^{\otimes n})$.
Equality follows by a dimension argument.
It is well-known from fusion category theory (see also \S\ref{sec:FusionCategoricalIdentification} below) that
$$
\dim\End_{\Hilb(G)}(\bbC[G]^{\otimes n})
=
\dim\End_{\Rep(G)}((\bbC^G)^{\otimes n}).
$$
We know that the dimension of $\fC(I) \cong \fD(I)$ is the same for both the rough and smooth boundary algebras from the bases in \eqref{eq:RoughCanonicalBasis} and \eqref{eq:SmoothCanonicalBasis} from Remark \ref{rem:DimensionsOfBoundaryAlgebras}.
As we already identified the rough boundary algebra in Theorem \ref{thm:2d-qd-matrix-decomp} above, the result follows.
\end{proof}

\subsubsection{Fusion categorical proof}
\label{sec:FusionCategoricalIdentification}

We now give a quick fusion categorical proof of Theorem \ref{thm:BoundaryAlgebrasForKQD} using Vaughan Jones' shaded planar algebras \cite{MR4374438}.

Consider the $2\times 2$ unitary multifusion category
$$
\begin{pmatrix}
\Hilb(G) & \Hilb
\\
\Hilb & \Rep(G)
\end{pmatrix}
$$
where the top right copy of $\Hilb$ is $\Mod-\bbC[G]$ in $\Hilb(G)$ generated by $X=\bbC[G]_{\bbC[G]}$,
and the bottom left copy of $\Hilb$ is $\bbC[G]-\Mod$ in $\Hilb(G)$ generated by $\overline{X}={}_{\bbC[G]}\bbC[G]$.
Observe that 
$$
X\otimes\overline{X} 
= 
\bbC[G]\otimes_{\bbC[G]}\bbC[G]
\cong
\bbC[G]
$$
as $G$-graded Hilbert spaces, and
$$
{}_{\bbC[G]}\overline{X} \otimes X_{\bbC[G]}
=
{}_{\bbC[G]}\bbC[G] \otimes \bbC[G]_{\bbC[G]}.
$$
Under the equivalence of categories between $\Bim(\bbC[G])=\bbC[G]-\Mod-\bbC[G]$ in $\Hilb(G)$ and $\Rep(G)$ \cite[\S5]{MR1868118}, the above space corresponds to $\bbC^G$.
In particular, observe that the algebra of $\bbC[G]-\bbC[G]$ bimodular endomorphisms of $\overline{X}\otimes X$ in $\Hilb(G)$
is spanned by the translation operators $T_g\coloneqq R_g^{-1}\otimes L_g$, which preserve the $G$-grading.

Vaughan Jones' shaded planar algebras \cite{MR4374438} give an elegant graphical representation of this multifusion category.
We denote $X$ as an unshaded-shaded strand, and $\overline{X}$ by its horizontal reflection:
\[
X
=
\tikzmath{
\begin{scope}
\clip[rounded corners=5pt] (-.5,-.5) rectangle (.5,.5);
\draw[dotted, rounded corners=5pt] (0,.5) -- (-.5,.5) -- (-.5,-.5) -- (0,-.5);
\fill[gray!30] (0,-.5) rectangle (.5,.5);
\end{scope}
\draw (0,-.5) -- (0,.5);
}
\qquad\qquad
\overline{X}
=
\tikzmath{
\begin{scope}
\clip[rounded corners=5pt] (-.5,-.5) rectangle (.5,.5);
\draw[dotted, rounded corners=5pt] (0,.5) -- (.5,.5) -- (.5,-.5) -- (0,-.5);
\fill[gray!30] (0,-.5) rectangle (-.5,.5);
\end{scope}
\draw (0,-.5) -- (0,.5);
}\,.
\]
Fusion products of $X$ and $\overline{X}$ are given by horizontal juxtaposition.  
Note that 
$\End(X \otimes \overline{X})\cong\End_{\Hilb(G)}(\bbC[G])$ is generated by the projections $P_g$ onto the simple objects,
and
$\End(\overline{X} \otimes X)\cong \End_{\Rep(G)}(\bbC^G)$ is generated by the translation operators $T_g$.
Graphically, we represent $P_g$ and $T_g$ as follows:
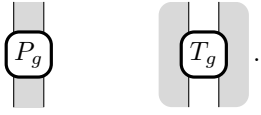
\[
\tikzmath{
\fill[gray!30](0, -.7) rectangle (.4, .7);
\draw (0, -.7) -- (0, .7);
\draw (.4, -.7) -- (.4, .7);
\roundNbox{fill=white}{(.2, 0)}{.3}{0}{0}{$P_g$}
}
\qquad\qquad
\tikzmath{
\begin{scope}
\clip[rounded corners=5pt](-.4, -.7) rectangle (.8, .7);
\fill[gray!30](-.4, -.7) rectangle (0, .7);
\fill[gray!30] (.4, -.7) rectangle (.8, .7);
\end{scope}
\draw (0, -.7) -- (0, .7);
\draw (.4, -.7) -- (.4, .7);
\roundNbox{fill=white}{(.2, 0)}{.3}{0}{0}{$T_g$}
}\,.
\]
We then have that the following exchange relation is satisfied \cite[\S7]{MR1950890}
(This shaded planar algebra is actually the standard invariant of the \emph{group subfactor} $R\subset R\rtimes G$; see \cite{MR1950890} for more details).
\[
\tikzmath{
\begin{scope}
\clip[rounded corners=5pt] (-1.4,-1) rectangle (1.4,1);
\fill[gray!30] (-1.4,-1) rectangle (-1,1);
\fill[gray!30] (-.6,-1) rectangle (-.2,1);
\fill[gray!30] (.2,-1) rectangle (.6,1);
\fill[gray!30] (1,-1) rectangle (1.4,1);
\end{scope}
\draw (-1,-1) -- (-1,1);
\draw (-.6,-1) -- (-.6,1);
\draw (-.2,-1) -- (-.2,1);
\draw (.2,-1) -- (.2,1);
\draw (.6,-1) -- (.6,1);
\draw (1,-1) -- (1,1);
\roundNbox{fill=white}{(0,.4)}{.3}{0}{0}{$T_g$}
\roundNbox{fill=white}{(.4,-.4)}{.3}{0}{0}{$P_h$}
}
=
\tikzmath{
\begin{scope}
\clip[rounded corners=5pt] (-1.4,-1) rectangle (1.4,1);
\fill[gray!30] (-1.4,-1) rectangle (-1,1);
\fill[gray!30] (-.6,-1) rectangle (-.2,1);
\fill[gray!30] (.2,-1) rectangle (.6,1);
\fill[gray!30] (1,-1) rectangle (1.4,1);
\end{scope}
\draw (-1,-1) -- (-1,1);
\draw (-.6,-1) -- (-.6,1);
\draw (-.2,-1) -- (-.2,1);
\draw (.2,-1) -- (.2,1);
\draw (.6,-1) -- (.6,1);
\draw (1,-1) -- (1,1);
\roundNbox{fill=white}{(0,-.4)}{.3}{0}{0}{$T_g$}
\roundNbox{fill=white}{(.4,.4)}{.3}{0}{0}{$P_{gh}$}
}
\]
For the rough boundary algebra, the operators $P_\ell^{g}$ and $Q_p^{(g)}$ satisfy these relations, with $P_\ell^{g}$ playing the role of $P_g$ and $Q_p^{(g)}$ playing the role of $T_g$.  
Hence, if $I$ is rough with $|I|=n$, 
\[
\fC(I)
\cong
\End_{\Hilb(G)}(\bbC[G]^{\otimes n}).
\]
Similarly, for the smooth algebra, $S_s^{(g)}$ and $R_\ell^{g}$ satisfy these relations.  
Hence if $I$ is smooth with $|I|=n$,
\[
\fC(I)
\cong
\End_{\Rep(G)}((\bbC^G)^{\otimes n}).
\]
This proves Theorem \ref{thm:BoundaryAlgebrasForKQD}.

Moreover, this 2-shaded proof also shows that both limit AF $\rmC^*$-algebras are isomorphic to $M_{|G|^\infty}$.
Indeed, $\varinjlim \End_{\Hilb(G)}(\bbC[G]^{\otimes n})$ is classified by the Bratteli diagram \cite{MR0112057,MR0312282,MR0397420}
\begin{equation}
\label{eq:BratteliHilbG}
\tikzmath{
\filldraw (0,0) node[left]{$\scriptstyle 1$} circle (.05cm);
\filldraw (.75,.5) node[above]{$\scriptstyle g_1$} circle (.05cm);
\node at (.75,.05) {$\vdots$};
\filldraw (.75,-.5) node[below]{$\scriptstyle g_k$} circle (.05cm);
\filldraw (1.5,0) circle (.05cm);
\filldraw (2.25,.5) node[above]{$\scriptstyle g_1$} circle (.05cm);
\node at (2.25,.05) {$\vdots$};
\filldraw (2.25,-.5) node[below]{$\scriptstyle g_k$} circle (.05cm);
\filldraw (3,0) circle (.05cm);
\node at (3.5,0) {$\cdots$};
\draw (0,0) -- (.75,-.5) -- (1.5,0) -- (2.25,-.5) -- (3,0);
\draw (0,0) -- (.75,.5) -- (1.5,0) -- (2.25,.5) -- (3,0);
}
\end{equation}
where $\{g_1,\dots, g_k\}$ is an enumeration of $G$.
Here, $\End_{\Hilb(G)}(\bbC[G]^{\otimes n})$ corresponds to the $k=|G|$ nodes of \eqref{eq:BratteliHilbG} on alternating levels
which correspond to the $|G|$ summands of size $|G|^{n-1}$ at depth $(n+1)/2$ (where we start counting at depth 0).
Similarly, $\varinjlim \End_{\Rep(G)}((\bbC^G)^{\otimes n})$ is is classified by the Bratteli diagram
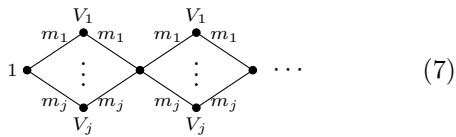
\begin{equation}
\label{eq:BratteliRepG}
\tikzmath{
\filldraw (0,0) node[left]{$\scriptstyle 1$} circle (.05cm);
\filldraw (.75,.5) node[above]{$\scriptstyle V_1$} circle (.05cm);
\node at (.75,.05) {$\vdots$};
\filldraw (.75,-.5) node[below]{$\scriptstyle V_j$} circle (.05cm);
\filldraw (1.5,0) circle (.05cm);
\filldraw (2.25,.5) node[above]{$\scriptstyle V_1$} circle (.05cm);
\node at (2.25,.05) {$\vdots$};
\filldraw (2.25,-.5) node[below]{$\scriptstyle V_j$} circle (.05cm);
\filldraw (3,0) circle (.05cm);
\node at (3.5,0) {$\cdots$};
\draw (0,0) --node[above]{$\scriptstyle m_1$} (.75,.5) --node[above]{$\scriptstyle m_1$} (1.5,0) --node[above]{$\scriptstyle m_1$} (2.25,.5) --node[above]{$\scriptstyle m_1$} (3,0);
\draw (0,0) --node[below]{$\scriptstyle m_j$} (.75,-.5) --node[below]{$\scriptstyle m_j$} (1.5,0) --node[below]{$\scriptstyle m_j$} (2.25,-.5) --node[below]{$\scriptstyle m_j$} (3,0);
}
\end{equation}
where $\{V_1,\dots, V_j\}$ is an enumeration of the irreps of $G$, where $m_i=\dim(V_i)$ for all $i$.
The result follows by the fact that $|G|=\sum_{i=1}^j m_j^2$ (and $|G|=k$).

\subsection{Canonical state on the boundary algebra}

In this section, we compute that the canonical state $\psi_\fB$ on $\fB = \varinjlim_I \fB(I)$ is a trace, connecting to \cite{2212.09036} on the type of cone von Neumann algebras.  

A \emph{cone} $\Lambda$ consists of all edges intersecting the region enclosed by two rays\footnote{The article \cite{2212.09036} uses a slightly different definition.  In her paper, she only considers edges that are completely contained in the region enclosed by the two rays.  However, she uses a dual convention to ours for the Quantum Double model, so this definition more closely aligns with hers when the differing conventions are taken into account.}, as illustrated below: 
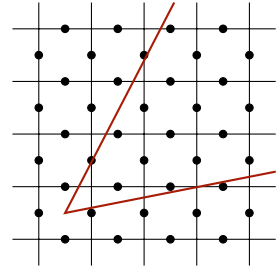
\[
\tikzmath{
\draw[scale = .7] (-2.5,-2.5) grid (2.5, 2.5);
\foreach \x in {-1.05, -.35, ..., 1.05}
{
\foreach \y in {-1.4, -.7, ..., 1.4}{
\filldraw (\x, \y) circle (0.05cm);
}
}
\foreach \y in {-1.05, -.35, ..., 1.05}
{
\foreach \x in {-1.4, -.7, ..., 1.4}{
\filldraw (\x, \y) circle (0.05cm);
}
}
\draw[thick, rufous] (.4, 1.75) -- (-1.05, -1.05) -- (1.75, -.5);
}
\]
We then define the algebra $\fA(\Lambda) \subseteq \fA$ to be the $\rmC^*$-subalgebra generated by all $\fA(\Delta)$ ranging over all rectangles $\Delta$ contained in $\Lambda$.  
By \cite{NaaijkensThesis}, the state $\psi$ on $\fA$ is the unique translation-invariant ground state for the Quantum Double model.  
Hence the \emph{cone algebras} studied in \cite{MR3426207,2212.09036} are given by $\pi_\psi(\fA(\Lambda))''$, where $\pi_\psi$ is the GNS representation for $\psi$.  
The article \cite{2212.09036} shows these von Neumann algebras are type $\rmI\rmI_\infty$\footnote{Ogata only claims this result in the case when the group is abelian.  However, her work in \cite{2212.09036}, combined with an argument from \cite{MR3426207}, can be used to prove the more general case.}.  
She did this essentially by showing that the state $\psi$ restricted to the compression of $\pi_\psi(\fA(\Lambda))''$ by the infimum of the projections $p_\Delta \in \fA(\Lambda)$ is a trace.
(Her approach exactly yields this result in the case where the cone is in fact a quadrant or a half-plane and has its vertex at the center of a face.)

In what follows, we consider the case where the cone is a half plane $\bbH$.  
In the case where the boundary of the half plane is rough, this is an example of a cone under the definition of \cite{2212.09036}.  
The more general case can be treated by adapting our work.  
We prove that $\pi_\psi(\fA_\bbH)''$ is a type $\rmI\rmI_\infty$ factor both when the boundary of $\bbH$ is rough and when the boundary of $\bbH$ is smooth.  
Indeed, $\fB = \varinjlim_I \fB(I)$ is exactly the compression of $\pi_\psi(\fA_{\bbH})''$ by the infimum $p$ of the projections $p_\Delta \in \fA_{\bbH}$, where $\bbH$ is the half-plane used in Definition \ref{defn:BoundaryAlgebras}.
Hence by proving that $\psi_\fB$ is a trace, we have that $p\pi_\psi(\fA_\bbH)''p$ is a type $\rmI\rmI_1$ factor, as the Bratteli diagram is stationary and periodic by (\ref{eq:BratteliHilbG},\ref{eq:BratteliRepG}).
The fact that this is true when the boundary of $\bbH$ is smooth is slightly surprising, since in that case $\fB(I) \cong \End_{\Rep(G)} ((\bbC^G)^{\otimes n})$.  
Indeed, the Levin-Wen model corresponding to $\Rep(G)$ studied in \cite{2307.12552} gives type $\rmI\rmI\rmI$ factors when $G$ is not abelian.

\begin{thm}
\label{thm:CanonicalStateTracial}
Let $\fB = \varinjlim_I \fB(I)$ as in Construction \ref{const:BoundaryAlgebra}.  
Then the canonical state $\psi_\fB$ on $\fB$ is tracial.
\end{thm}

\begin{proof} 
It suffices to show that for each interval $I$, $\psi_\fB(xy) = \psi_\fB(yx)$ for all $x, y \in \fB(I)$.  
The map $\fC(I) \to \fB(I)$ given by $x \mapsto x p_{\Delta_I}$ is an isomorphism, and for all $x \in \fC(I)$
\[
\psi_\fB(xp_{\Delta_I})
=
\psi(x p_{\Delta_I})
=
\psi(x).
\]
Hence it suffices to show that $\psi(xy) = \psi(yx)$ for all $x, y \in \fC(I)$.
We prove this result in the case that $I$ is smooth; the proof when $I$ is rough is analogous.

Let $n$ denote the number of edges in $I$.
Then a canonical basis of $\fC(I)$ is given by $\set{\prod_{i = 1}^{n}R_{\ell_i}^{g_i}\prod_{i = 1}^{n-1} S_{s_i}^{(h_i)}}{g_i,h_i\in G}$. 
It suffices to prove that $\psi(xy) = \psi(yx)$ for basis elements $x,y$.

Let $x = \prod_{i = 1}^{n}R_{\ell_i}^{g_i}\prod_{i = 1}^{n-1} S_{s_i}^{(h_i)}$ for some $R_{\ell_i}^{g_i}$ and $S_{s_i}^{(h_i)}$. 
Observe that
$$
\psi(x) =\begin{cases}
		    \frac{1}{|G|^{n-1}}, & \text{if $g_i = e$ for all $i$}\\
            0, & \text{otherwise.}
		 \end{cases}
$$
Now, we let $x = \prod_{i = 1}^{n}R_{\ell_i}^{g_i}\prod_{i = 1}^{n-1} S_{s_i}^{(h_i)}$ and $y = \prod_{i = 1}^{n}R_{\ell_i}^{\widetilde{g}_i}\prod_{i = 1}^{n-1} S_{s_i}^{(\widetilde{h}_i)}$.  
Then by the exchange relation
\[
S_{s_i}^{(h_i)}R_{\ell_i}^{\widetilde{g}_i}R_{\ell_{i+1}}^{\widetilde{g}_{i+1}} = R_{\ell_i}^{\widetilde{g}_i}R_{\ell_{i+1}}^{\widetilde{g}_{i+1}}S_{s_i}^{(\widetilde{g}_ih_i\widetilde{g}_{i+1}^{-1})}
\]
we have that 
\[
\psi(xy) = \prod_{i = 1}^{n}R_{\ell_i}^{g_i\widetilde{g}_i}\prod_{i = 1}^{n-1}(S_{s_i}^{(\widetilde{g}_ih_i\widetilde{g}_{i+1}^{-1})}S_{s_i}^{(\widetilde{h}_i)})
\]
Hence  we have that $\psi(xy) = \frac{1}{|G|^{n-1}}$ if $g_i\widetilde{g_i} = e$ and $\widetilde{g}_ih_i\widetilde{g}_{i+1}^{-1} = \widetilde{h}_i$ for every $i$, and $\psi(xy) = 0$ otherwise.
By the same argument, $\psi(yx) = \frac{1}{|G|^n}$ under the same conditions and $\psi(yx) = 0$ otherwise.  
Hence $\psi(xy) = \psi(yx)$.
\end{proof}

\begin{cor}
Let $\bbH$ be a half plane with rough or smooth boundary.  
Then $\pi_\psi(\fA_\bbH)''$ is a type $\rmI\rmI_\infty$ factor, where $\pi_\psi$ is the GNS representation for $\psi$.  
\end{cor}

\begin{proof}
By arguments in \cite{MR3426207}, $\pi_\psi(\fA_\bbH)''$ is an infinite factor.  
Let $p \coloneqq \bigwedge_{\Delta \subseteq \bbH} p_\Delta$.  
Then since $\fB \cong p \pi_\psi(\fA_\bbH)'' p$, we have by Theorem \ref{thm:CanonicalStateTracial} that $p \pi_\psi(\fA_\bbH)''p$ is a type $\rmI\rmI_1$ factor.
The result follows.
\end{proof}

\section{Isomorphism with the SPT Boundary Algebra}\label{sec:SPT Boundary Algebra}

In addition to our direct method of calculating the boundary algebra of the Quantum Double model, we may calculate it using gauge theory.  
One construction of the Quantum Double model for a group $G$ is by gauging a trivial bosonic $G$-Symmetry Protected Topological theory (SPT). 
The gauging map induces a relationship between the boundary algebra of the Quantum Double and certain operators on the $G$-equivariant SPT Hilbert space on a finite lattice.

We first construct the $G$-SPT and then review the standard gauging procedure to produce the Quantum Double \cite{PhysRevB.87.125114}. 
We then prove a correspondence between the low energy operators of the SPT and the boundary algebra of the Quantum Double.  
Finally, we will compute the low energy operators of the SPT, thereby computing the boundary algebra of the Quantum Double.

Let $G$ be a finite group and consider the $\Z^2$ lattice in the plane, denoted $\cL$. 
We associate $\mathbb{C}^G$-spins to each vertex of the lattice, rather than each edge.
Since the edges play no role in the SPT-theory until we construct the gauging map, we draw edges in orange rather than black.
$$
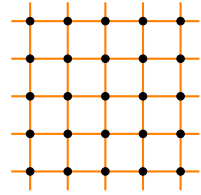
\tikzmath{
\draw[thick, orange, scale = .5] (-2.5,-2.5) grid (2.5, 2.5);
\foreach \x in {-1,-.5,0,.5,1}{
\foreach \y in {-1,-.5,0,.5,1}{
\filldraw (\x, \y) circle (0.05cm);
}}
}
$$
For $g\in G$, we use the operators $L^g_v,R^g_v,P^g_v$ on each vertex $v$ as they are defined on the edges of the Quantum Double.

\begin{nota}
For any rectangular subgraph $\Delta\subset \mathcal{L}$, we obtain a finite-dimensional Hilbert space $\calH_{\spt}=\cH_{\spt}(\Delta) = \C[G]^{\otimes N}$, where $N$ is the number of vertices in $\Delta$. We will consider two types of rectangular regions $\Delta$, rough and smooth, respectively, as in the following diagram:
$$
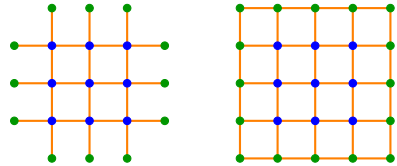
\tikzmath{
\foreach \y in {0,.5,1} {
\draw[thick,orange] (-.5,\y) -- (1.5,\y);
}
\foreach \x in {0,.5,1} {
\draw[thick,orange] (\x,-.5) -- (\x,1.5);
}
\foreach \x in {0,.5,1} {
\filldraw[DarkGreen] (\x,-.5) circle (.05cm);
\filldraw[DarkGreen] (\x,1.5) circle (.05cm);
\foreach \y in {0,.5,1} {
\filldraw[blue] (\x,\y) circle (.05cm);
}}
\foreach \y in {0,.5,1} {
\filldraw[DarkGreen] (-.5,\y) circle (.05cm);
\filldraw[DarkGreen] (1.5,\y) circle (.05cm);
}
\foreach \y in {-.5,0,.5,1,1.5} {
\draw[thick,orange] (2.5,\y) -- (4.5,\y);
}
\foreach \x in {2.5,3,3.5,4,4.5} {
\draw[thick,orange] (\x,-.5) -- (\x,1.5);
}
\foreach \x in {0,.5,1} {
\filldraw[DarkGreen] (\x+3,-.5) circle (.05cm);
\filldraw[DarkGreen] (\x+3,1.5) circle (.05cm);
\foreach \y in {0,.5,1} {
\filldraw[blue] (\x+3,\y) circle (.05cm);
}}
\foreach \y in {0,.5,1} {
\filldraw[DarkGreen] (2.5,\y) circle (.05cm);
\filldraw[DarkGreen] (4.5,\y) circle (.05cm);
}
\filldraw[DarkGreen] (4.5,1.5) circle (.05cm);
\filldraw[DarkGreen] (2.5,1.5) circle (.05cm);
\filldraw[DarkGreen] (4.5,-.5) circle (.05cm);
\filldraw[DarkGreen] (2.5,-.5) circle (.05cm);
}
$$
We let $\Delta^\circ$ denote the full subgraph of $\Delta$ of all vertices of degree four in $\Delta$, i.e., all vertices whose 4 nearest neighbors all lie in $\Delta$.
Above, we denote vertices in $\Delta^\circ$ by blue nodes.
We define $\partial \Delta\coloneqq \Delta\setminus \Delta^\circ$, which we denote above by green nodes. 
(When we study boundary algebras for rough edges below, some green nodes on the edges will be shaded gray; we will discus this distinction later on.)
\end{nota}

We define the Hamiltonian on $\calH_{\spt}$ by
\[
    H_{\spt}=H_{\spt}(\Delta)\coloneqq
    -\sum\limits_{v\in\Delta^\circ}B_v
\]
where $B_v$ is the orthogonal projector on $\bbC^G$ given by 
\[
B_v = \frac{1}{|G|} \sum_{g \in G} R^g_v.
\]
There is a symmetry action on $\calH_{\spt}$ defined by the map $U_g = \prod_{v}L^g_v$ for each $g \in G$. 
Now consider the $G$-equivariant subspace $\calH^G_{\spt} \subset \calH_{\spt}$, which is the space of fixed points of the $U_g$ for all $g \in G$, i.e., 
\[
\cH_{\spt}^G
\coloneqq 
\set{|v\rangle \in \cH_{\spt}}{U_g|v\rangle = |v\rangle \; \forall g \in G}.
\]
Notice that $P_G \coloneqq \frac{1}{|G|}\sum_{g \in G} U_g$ is the projector onto $\cH_{\spt}^G$,
and
\[
\set{\frac{1}{\sqrt{|G|}}\sum_{g \in G} |g, gh_1, \ldots, gh_{N-1}\rangle}{h_i \in G}
\]
is an orthonormal basis for $\cH_{\spt}^G$.

We now define the gauging map $\Gamma\colon\cH_{\spt}^G \to \calH_{\qd}$, with both systems on finite lattices. 
The edges in the Quantum Double lattice will be the edges of the dual lattice of $\mathcal{L}$ which intersect an edge in $\Delta$. 
The intersection points exactly correspond to the sites of the Quantum Double.
Importantly, the rough boundary in the SPT will become the smooth boundary in the Quantum Double and vice versa.
The black lines represent the region $\Delta$ for the Quantum Double as defined in the earlier sections.
Here is an example below building on our earlier examples.
$$
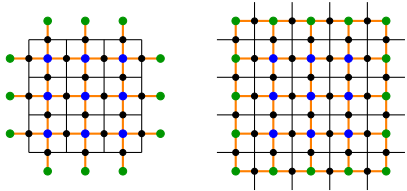
\tikzmath{
\foreach \y in {0,.5,1} {
\draw[thick,orange] (-.5,\y) -- (1.5,\y);
}
\foreach \x in {0,.5,1} {
\draw[thick,orange] (\x,-.5) -- (\x,1.5);
}
\foreach \y in {-.25,.25,.75,1.25} {
\draw[black] (-.25,\y) -- (1.25,\y);
}
\foreach \x in {-.25,.25,.75,1.25} {
\draw[black] (\x,-.25) -- (\x,1.25);
}
\foreach \x in {-.25,.25,.75,1.25}{
\foreach \y in {0,.5,1}{
\filldraw[black] (\x,\y) circle (.04cm);
}
}

\foreach \x in {0,.5,1}{
\foreach \y in {-.25,.25,.75,1.25}{
\filldraw[black] (\x,\y) circle (.04cm);
}
}

\foreach \x in {0,.5,1} {
\filldraw[DarkGreen] (\x,-.5) circle (.05cm);
\filldraw[DarkGreen] (\x,1.5) circle (.05cm);
\foreach \y in {0,.5,1} {
\filldraw[blue] (\x,\y) circle (.05cm);
}}
\foreach \y in {0,.5,1} {
\filldraw[DarkGreen] (-.5,\y) circle (.05cm);
\filldraw[DarkGreen] (1.5,\y) circle (.05cm);
}
\foreach \y in {-.5,0,.5,1,1.5} {
\draw[thick,orange] (2.5,\y) -- (4.5,\y);
}
\foreach \x in {2.5,3,3.5,4,4.5} {
\draw[thick,orange] (\x,-.5) -- (\x,1.5);
}

\foreach \y in {-.25,.25,.75,1.25} {
\draw[black] (2.25,\y) -- (4.75,\y);
}
\foreach \x in {2.75,3.25,3.75,4.25} {
\draw[black] (\x,-.75) -- (\x,1.75);
}
\foreach \x in {2.75,3.25,3.75,4.25}{
\foreach \y in {-.5,0,.5,1,1.5}{
\filldraw[black] (\x,\y) circle (.04cm);
}
}

\foreach \x in {2.5,3,3.5,4,4.5}{
\foreach \y in {-.25,.25,.75,1.25}{
\filldraw[black] (\x,\y) circle (.04cm);
}
}

%
\foreach \x in {0,.5,1} {
\filldraw[DarkGreen] (\x+3,-.5) circle (.05cm);
\filldraw[DarkGreen] (\x+3,1.5) circle (.05cm);
\foreach \y in {0,.5,1} {
\filldraw[blue] (\x+3,\y) circle (.05cm);
}}
\foreach \y in {0,.5,1} {
\filldraw[DarkGreen] (2.5,\y) circle (.05cm);
\filldraw[DarkGreen] (4.5,\y) circle (.05cm);
}
\filldraw[DarkGreen] (4.5,1.5) circle (.05cm);
\filldraw[DarkGreen] (2.5,1.5) circle (.05cm);
\filldraw[DarkGreen] (4.5,-.5) circle (.05cm);
\filldraw[DarkGreen] (2.5,-.5) circle (.05cm);
}
$$

One can interpret this construction as a derivation of the Quantum Double from the SPT.  
As we will see, the map $\Gamma$ induces a map on operators which sends the SPT Hamiltonian to the flux ($B_p$) terms in the Quantum Double Hamiltonian.  
The local gauge constraint obtained from the gauging map will be enforced in the Quantum Double by adding the star terms ($A_v$).

We orient the edges of $\mathcal{L}$ from left to right and up to down.
Consider the map from $\cH_{\spt}$ to $\cH_{\qd}$ 
which maps between standard ONB vectors according to the following rule.
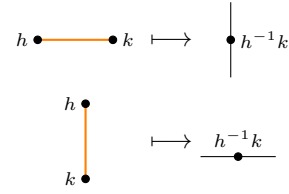
\begin{align*}
\tikzmath{
\draw[thick, orange] (0,0) -- (1,0);
\filldraw (0,0)node[left]{$\scriptstyle h$} circle (.05cm);
\filldraw (1,0)node[right]{$\scriptstyle k$} circle (.05cm);
}
&\longmapsto
\quad
\tikzmath{
\draw (0,0) -- (0,1);
\filldraw (0,.5) node[right]{$\scriptstyle h^{-1}k$} circle (.05cm);
}
\\
\tikzmath{
\draw[thick, orange] (0,0) -- (0,1);
\filldraw (0,0)node[left]{$\scriptstyle k$} circle (.05cm);
\filldraw (0,1)node[left]{$\scriptstyle h$} circle (.05cm);
}
\qquad
&\longmapsto
\tikzmath{
\draw (0,0) -- (1,0);
\filldraw (.5,0) node[above]{$\scriptstyle h^{-1}k$} circle (.05cm);
}
\end{align*}
Observe that this map does not see the $U_g$ action on $\cH_{\spt}$, as 
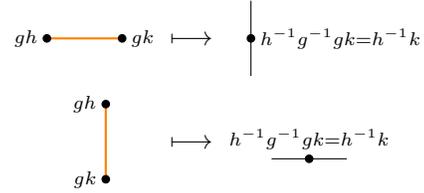
\begin{align*}
\tikzmath{
\draw[thick, orange] (0,0) -- (1,0);
\filldraw (0,0)node[left]{$\scriptstyle gh$} circle (.05cm);
\filldraw (1,0)node[right]{$\scriptstyle gk$} circle (.05cm);
}
&\longmapsto
\quad
\tikzmath{
\draw (0,0) -- (0,1);
\filldraw (0,.5) node[right]{$\scriptstyle h^{-1}g^{-1}gk=h^{-1}k$} circle (.05cm);
}
\\
\tikzmath{
\draw[thick, orange] (0,0) -- (0,1);
\filldraw (0,0)node[left]{$\scriptstyle gk$} circle (.05cm);
\filldraw (0,1)node[left]{$\scriptstyle gh$} circle (.05cm);
}
\qquad
&\longmapsto
\tikzmath{
\draw (0,0) -- (1,0);
\filldraw (.5,0) node[above]{$\scriptstyle h^{-1}g^{-1}gk=h^{-1}k$} circle (.05cm);
}
\end{align*}
We thus see that this map descends to a map on $P_G\cH_{\spt}$ given by
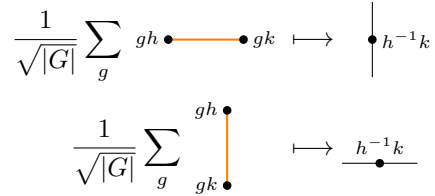
\begin{align*}
\frac{1}{\sqrt{|G|}}\sum_g
\tikzmath{
\draw[thick, orange] (0,0) -- (1,0);
\filldraw (0,0)node[left]{$\scriptstyle gh$} circle (.05cm);
\filldraw (1,0)node[right]{$\scriptstyle gk$} circle (.05cm);
}
&\longmapsto
\quad
\tikzmath{
\draw (0,0) -- (0,1);
\filldraw (0,.5) node[right]{$\scriptstyle h^{-1}k$} circle (.05cm);
}
\\
\frac{1}{\sqrt{|G|}}\sum_g
\tikzmath{
\draw[thick, orange] (0,0) -- (0,1);
\filldraw (0,0)node[left]{$\scriptstyle gk$} circle (.05cm);
\filldraw (0,1)node[left]{$\scriptstyle gh$} circle (.05cm);
}
\qquad
&\longmapsto
\tikzmath{
\draw (0,0) -- (1,0);
\filldraw (.5,0) node[above]{$\scriptstyle h^{-1}k$} circle (.05cm);
}
\end{align*}

That gauging the $G$-SPT produces the Quantum Double model follows from the following theorem.

\begin{thm}
\label{thm:main_SPT}
The gauging map $\Gamma$ is a unitary isomorphism 
$\cH_{\spt}^G \to \im(\prod_{s\subset \Delta^\circ} A_s)\subset \cH_{\qd}$.
Moreover, conjugation by $\Gamma$ sends
$$
-\sum_{v\in \Delta^\circ} B_v
\overset{\Ad(\Gamma)}{\longmapsto}
-\sum_{p_v}B_{p_v}
$$
\end{thm}
\begin{proof}
The proof of the first claim consists of three steps:
(1) $\im(\Gamma) \subset \im(\prod_{s\subset \Delta^\circ} A_s)$,
(2) $\Gamma$ is isometric, and 
(3), $\Gamma$ is surjective onto $\im(\prod_{s\subset \Delta^\circ} A_s)$.

To prove (1),
we note that each star $s\subset \Delta^\circ$ is the image of a face in the SPT.
Notice that $\Gamma$ maps a face as follows,
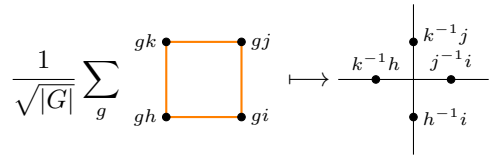
\[
\frac{1}{\sqrt{|G|}}\sum_g
\tikzmath{
\draw[orange, thick] (0,0) rectangle (1,1);
\filldraw (0,0)node[left]{$\scriptstyle gh$} circle (0.05);
\filldraw (1,0)node[right]{$\scriptstyle gi$} circle (0.05);
\filldraw (1,1)node[right]{$\scriptstyle gj$} circle (0.05);
\filldraw (0,1)node[left]{$\scriptstyle gk$} circle (0.05);
} 
\longmapsto
\tikzmath{
\draw (-1,0) -- (1,0);
\draw (0,-1) -- (0,1);
\filldraw (.5,0)node[above]{$\scriptstyle j^{-1}i$} circle (0.05);
\filldraw (0,.5)node[right, yshift=.1cm]{$\scriptstyle k^{-1}j$} circle (0.05);
\filldraw (-.5,0)node[above]{$\scriptstyle k^{-1}h$} circle (0.05);
\filldraw (0,-.5)node[right]{$\scriptstyle h^{-1}i$} circle (0.05);
}
\]
and the star on the right is in the image of the corresponding $A_s$ operator.

By construction, $\Gamma$ sends our chosen normalized elements of $\calH_{\spt}^G$ to the basis elements of $\calH_{\qd}$. 
Moreover, it maps distinct ONB elements to distinct ONB elements, because the difference in group element labels on two sites in $\calH_{\spt}^G$ is the same as the product of the path of edge labels between these two sites on the dual lattice in $\calH_{\qd}$, being careful to invert factors in this product based on the orientation of the edges.
Note that $\Delta$ is simply connected in the sense that any closed path may be deformed to the identity by shrinking the path across one plaquette at a time.
Since $\Delta$ is simply connected and $\im(\Gamma)\subseteq \im\left(\prod_{s \subset \Delta^\circ} A_s \right)$, the resulting ONB element is independent of the choice of path.
Physically, this corresponds to the fact that $G$-equivariant observables in the SPT are also observables in the Quantum Double.
Hence $\Gamma$ is isometric, proving (2).

To prove (3), pick any basis state $|v\rangle\in\im\left(\prod_{s \subset \Delta^\circ} A_s \right)$, and 
pick a maximal spanning tree 
$\Lambda\subset \Delta$.
Since $\Lambda$ is acyclic, we may choose a basis state $|w\rangle\in\calH_{\spt}^G$ such that the basis state $\Gamma|w\rangle$ agrees with $|v\rangle$ on the edges of the dual lattice which $\Lambda$ intersects. 
By using the fact that $|v\rangle \in\im\left(\prod_{s \subset \Delta^\circ} A_s \right)$ we can deduce that $\Gamma|w\rangle$ agrees with $|v\rangle$ on the remaining edges as well. 

To prove the second claim, using the following diagram, one can derive the fact that the operator $R_v^h$ in the SPT gets sent to $B^{h^{-1}}_{p_v}$ where $p_v$ is the plaquette in the Quantum Double associated with the vertex $v$.
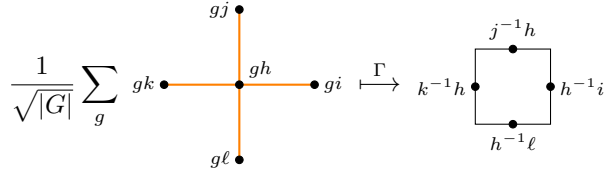
$$
\frac{1}{\sqrt{|G|}}\sum_g
\tikzmath{
\draw[thick, orange] (0,-1) -- (0,1);
\draw[thick, orange] (-1,0) -- (1,0);
\filldraw (0,0) node[right, yshift=.2cm]{$\scriptstyle gh$} circle (.05cm);
\filldraw (1,0) node[right]{$\scriptstyle gi$} circle (.05cm);
\filldraw (0,1) node[left]{$\scriptstyle gj$} circle (.05cm);
\filldraw (-1,0) node[left]{$\scriptstyle gk$} circle (.05cm);
\filldraw (0,-1) node[left]{$\scriptstyle g\ell$} circle (.05cm);
}
\overset{\Gamma}{\longmapsto}
\tikzmath{
\draw (-.5,-.5) rectangle (.5,.5);
\filldraw (.5,0)node[right]{$\scriptstyle h^{-1}i$} circle (.05cm);
\filldraw (0,.5)node[above]{$\scriptstyle j^{-1}h$} circle (.05cm);
\filldraw (-.5,0)node[left]{$\scriptstyle k^{-1}h$} circle (.05cm);
\filldraw (0,-.5)node[below]{$\scriptstyle h^{-1}\ell$} circle (.05cm);
}
$$
Therefore conjugation by $\Gamma$ sends the SPT Hamiltonian $-\sum_{v\in \Delta^\circ} B_v$
on $\cH_{\spt}^G$
to $-\sum_{p_v}B_{p_v}$ on $\cH_{\qd}$ as claimed.
\end{proof}

It remains to describe the boundary algebra of $\calH^G_{\spt}$.
We take $\Lambda$ to be a subrectangle in the rectangular region $\Delta$. 
We also require that $\partial\Delta\cap\partial\Lambda$ is some 1D interval of vertices $I$;
here, $\partial\Delta$ and $\partial\Lambda$ are the outermost vertices of $\Delta$ and $\Lambda$ respectively.

We define the boundary algebra of the SPT to be the algebra of operators $p_\Delta x p_\Delta $ where $x$ is a $G$-equivariant operator supported in $\Lambda$ and $p_\Delta$ is the product of $B_v$ operators with $v\in\Delta^\circ$. 
These are the $G$-equivariant low-energy operators where we equate two operators if they are equivalent on the low-energy subspace.

We demonstrate the geometry of the interval $I$ as follows with the smooth boundary of the Quantum Double on the left and the rough boundary on the right.  The boundary algebra of the Quantum Double is supported on the edges intersected by the solid orange lines.
We will explain the meaning of the gray vertices later.

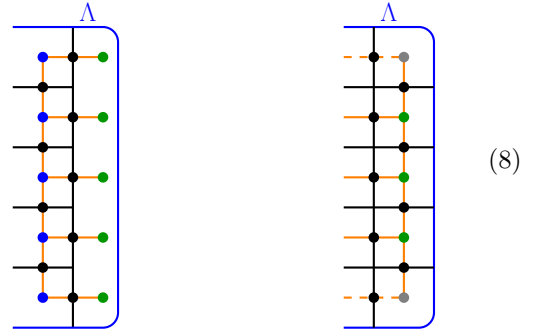
\begin{equation}
\label{eq:RoughAndSmoothSPTvsQD}
\tikzmath[auto=center,every node/.style={circle, fill=black, scale=0.4}, thick, scale=.8]{
    \draw (0, 4) -- (0, -1);
    \draw (0, 3) -- (-1, 3);
    \draw (0, 2) -- (-1, 2);
    \draw (0, 1) -- (-1, 1);
    \draw (0, 0) -- (-1, 0);
    \draw[orange] (-.5,-.5)--(-.5,3.5); 
    \draw[orange] (-.5,-.5)--(.5,-.5);
    \draw[orange] (-.5,.5)--(.5,.5);
    \draw[orange] (-.5,1.5)--(.5,1.5);
    \draw[orange] (-.5,2.5)--(.5,2.5);
    \draw[orange] (-.5,3.5)--(.5,3.5);
    \filldraw[DarkGreen] (.5,-.5) circle (.07cm);
    \filldraw[DarkGreen] (.5,.5) circle (.07cm);
    \filldraw[DarkGreen] (.5,1.5) circle (.07cm);
    \filldraw[DarkGreen] (.5,2.5) circle (.07cm);
    \filldraw[DarkGreen] (.5,3.5) circle (.07cm);
    \filldraw[blue] (-.5,-.5) circle (.07cm);
    \filldraw[blue] (-.5,.5) circle (.07cm);
    \filldraw[blue] (-.5,1.5) circle (.07cm);
    \filldraw[blue] (-.5,2.5) circle (.07cm);
    \filldraw[blue] (-.5,3.5) circle (.07cm);
    \filldraw[black] (0,-.5) circle (0.07cm);
    \filldraw[black] (0,.5) circle (0.07cm);
    \filldraw[black] (0,1.5) circle (0.07cm);
    \filldraw[black] (0,2.5) circle (0.07cm);
    \filldraw[black] (0,3.5) circle (0.07cm);
    \filldraw[black] (-.5,0) circle (0.07cm);
    \filldraw[black] (-.5,1) circle (0.07cm);
    \filldraw[black] (-.5,2) circle (0.07cm);
    \filldraw[black] (-.5,3) circle (0.07cm);
    \draw[blue] (-1, 4) -- (0.5,4) 
        arc [x radius = 0.25, y radius = 0.25, start angle = 90, end angle = 0] -- (0.75,-0.75) arc [x radius = 0.25, y radius = 0.25, start angle = 0, end angle = -90] -- (-1, -1);
    
    \draw[orange] (5.5,-.5)--(5.5,3.5);
    \draw[dashed,orange] (4.5,-.5)--(5.5,-.5);
    \draw[orange] (4.5,.5)--(5.5,.5);
    \draw[orange] (4.5,1.5)--(5.5,1.5);
    \draw[orange] (4.5,2.5)--(5.5,2.5);
    \draw[dashed,orange] (4.5,3.5)--(5.5,3.5);
    \filldraw[gray] (5.5,-.5) circle (.07cm);
    \filldraw[DarkGreen] (5.5,.5) circle (.07cm);
    \filldraw[DarkGreen] (5.5,1.5) circle (.07cm);
    \filldraw[DarkGreen] (5.5,2.5) circle (.07cm);
    \filldraw[gray] (5.5,3.5) circle (.07cm);
    \filldraw[black] (5,-.5) circle (0.07cm);
    \filldraw[black] (5,.5) circle (0.07cm);
    \filldraw[black] (5,1.5) circle (0.07cm);
    \filldraw[black] (5,2.5) circle (0.07cm);
    \filldraw[black] (5,3.5) circle (0.07cm);
    \filldraw[black] (5.5,0) circle (0.07cm);
    \filldraw[black] (5.5,1) circle (0.07cm);
    \filldraw[black] (5.5,2) circle (0.07cm);
    \filldraw[black] (5.5,3) circle (0.07cm);
    \draw[blue] (4.5, 4) -- (5.75,4) 
        arc [x radius = 0.25, y radius = 0.25, start angle = 90, end angle = 0] -- (6,-0.75) arc [x radius = 0.25, y radius = 0.25, start angle = 0, end angle = -90] -- (4.5, -1);
    \draw (5, 4) -- (5, -1);
    \draw (4.5, 3) -- (6, 3);
    \draw (4.5, 2) -- (6, 2);
    \draw (4.5, 1) -- (6, 1);
    \draw (4.5, 0) -- (6, 0);
    \node[fill=none] at (0.25, 4.25) {\textcolor{blue}{\Huge{$\Lambda$}}};
    \node[fill=none] at (5.25, 4.25) {\textcolor{blue}{\Huge{$\Lambda$}}};
}
\end{equation}

A complete set of basis ground states of this SPT are given by product states on each site such that all sites in $\Delta^\circ$ have the state $\frac{1}{\sqrt{|G|}}\sum\limits_{g\in G}|g\rangle$. Hence the low-energy operators are all supported on $\partial\Delta=\Delta\setminus\Delta^\circ$. 
In other words, all low-energy operators supported in $\Lambda$ are supported on the sites in $I=\partial \Delta\cap \partial \Lambda$. 
Furthermore, the algebra of all low energy operators supported in $\Lambda$ is the algebra of all $G$-equivariant operators supported on the sites in $I$.
Thus the boundary algebra of the SPT is $\End_{\Rep(G)}((\bbC^G)^{\otimes n})$ where $n=|I|$. 

We now describe a generating set for the $G$-equivariant operators supported on $I$ when $I$ is a rough boundary of the SPT/smooth boundary of the Quantum Double.
For each $k \in G$, define an operator $Z^k_{i,j}$, that acts as the identity everywhere except at the sites $i,j$, where it is defined on basis elements as 
$$
Z^k_{i,j} |g_i, g_j \rangle \coloneqq \delta_{k= g_i^{-1}g_j} |g_i, g_j \rangle.
$$
It is easy to check that $Z^k_{i,j}$ is $G$-equivariant for all $k \in G$. 
Notice that the set 
\[
\set{R^g_v , Z^k_{i,j}}{g, k \in G , \; v, i, j \in V(\mathcal{L})}
\] 
is a generating set for the algebra of $G$-equivariant operators on $\calH_{\spt}$.

Since $\Gamma$ maps $\cH_{\spt}^G$ to $\im(p_A)\subset \cH_{\qd}$,
under conjugation by $\Gamma$ with adjacent vertices $i,i+1$ in $I$ with $i+1$ just below $i$, $Z^k_{i,i+1}$ is sent to 
the operator which projects the ghost edge $(i,i+1)$ connecting $i$ and $i+1$ to $|k\rangle$;
this is exactly the operator $S^{(k)}_{s\setminus(i,i+1)}$ where $s\setminus(i,i+1)$ is the partial
star in $\Delta$ with ghost edge $(i,i+1)$ on the right.
Conjugation by $\Gamma$ also sends $R^g_v$ to $R^g_\ell$ where $\ell$ is the edge just to the left of $v$.
We immediately obtain the following corollary.

\begin{cor}\label{cor:SPTBoundaryAlg}
When $I$ is a rough boundary of the SPT/smooth boundary of the Quantum Double,
the algebra of $G$-equivariant operators supported on $I$ is isomorphic to the smooth boundary algebra for the Quantum Double model under conjugation by the gauging unitary $\Gamma$.
\end{cor}

We now focus on the case when $I$ is a smooth boundary of the SPT/rough boundary of the Quantum Double.
Here, we need to be slightly more careful about the support of the operators in the boundary algebra for the Quantum Double.
Observe that the $\Gamma$-conjugates of the
$Z_{i,j}^k$ for any two vertices $i,j\in I$ are supported on edges in the Quantum Double boundary,
as are all $\Gamma$-conjugates of the $R_v^g$ supported on the green vertices in the diagram on the right hand side of \eqref{eq:RoughAndSmoothSPTvsQD}. 
However, the $\Gamma$-conjugates of the $R_v^g$ supported on the two outside gray vertices has support not contained within $\Lambda$!
Hence excluding these  $R_v^g$ supported on the two outside gray vertices, conjugation by $\Gamma$ again gives an explicit isomorphism to the Quantum Double boundary algebra for such $I$.

\begin{cor}\label{cor:SPTBoundaryAlg2}
When $I$ is a smooth boundary of the SPT/rough boundary of the Quantum Double,
a distinguished subalgebra of the $G$-equivariant operators supported on $I$ is isomorphic to the rough boundary algebra for the Quantum Double model under conjugation by the gauging unitary $\Gamma$.
\end{cor}

For completeness, we include another proof below that this distinguished subalgebra is isomorphic to $\End_{\Hilb(G)}(\bbC[G]^{\otimes n})$ using the SPT boundary algebra generators.

Call the upper most gray vertex $t$ and the bottom most gray vertex $b$.
Notice that no element of the boundary algebra can mix the distinct eigenspaces of the collection of $Z_{t,b}^g$ operators.
This imposes a $G$-grading on the invariant subspaces of the action of the boundary algebra.
In particular, each of these spaces has dimension at most $|G|^n$ where $n=|I|-1$, since we have a $G$-grading on a $|G|^{|I|}$ dimensional boundary Hilbert space.
In fact, each of these spaces is of dimension $|G|^n$.
To see this, notice that the operators $R^{(g)}_v$ can map all eigenspaces of the $n$ independent $Z_{i,i+1}^g$ operators for $b\leq i<t$ into one another.
The $Z_{i,i+1}^g$ operators may project onto any such eigenspace as well.
Therefore, all states within each graded component may be mapped to one another by the boundary algebra.
Therefore, this algebra is isomorphic to $\End_{\Hilb(G)}(\bbC[G]^{\otimes n})$ where $n=|I|-1$. 
(This decomposition of  $\End_{\Hilb(G)}(\bbC[G]^{\otimes n})$ follows immediately from \eqref{eq:BratteliHilbG}.)

It is important to note that the above computation of the boundary algebras in Corollaries \ref{cor:SPTBoundaryAlg} and \ref{cor:SPTBoundaryAlg2} using Theorem \ref{thm:main_SPT} is not limited to the particular type of SPT we chose.  
Indeed, had we started with a non-trivial SPT, gauging would have produced a twisted Quantum Double model and we could have found the boundary algebra in that case.

Moreover, our construction does not rely on the geometry of the lattice we chose.  That is, if we generalize to higher dimensions or non-cubical lattices, we could do a similar construction. Lastly, by computing boundary algebras in this way, we may even consider mixed or disordered boundary conditions.

\section{3D Quantum Double Boundary Algebra}

In this section, we show how the techniques used in Sections \ref{sec: Boundary Algebra} and \ref{sec:SPT Boundary Algebra} extend naturally to the computation of the boundary algebra for the three-dimensional Quantum Double model. 
As noted earlier, this can be done for the general $n$-dimensional Quantum Double model, but for this section we focus on the three-dimensional case for simplicity.
Let $G$ be an abelian group. We show that the boundary algebra admits a nice description in terms of characters of the irreducible representations of $G$.

For this model, we consider a 3D edge lattice $\mathcal{L} \subset \mathbb{R}^3$ where each edge carries $\C[G]$ spins. For a finite region $\Delta \subset \mathcal{L}$, we obtain the finite-dimensional Hilbert space $\calH = \bigotimes_{\ell \in \Delta} \C[G]$. 
The edges of this lattice are oriented as follows.
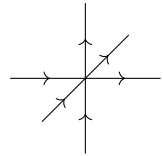
\[
\tikzmath{
\draw[mid>] (0,0,0) -- (1,0,0);
\draw[mid>] (-1,0,0) -- (0,0,0);
\draw[mid>] (0,-1,0) -- (0,0,0);
\draw[mid>] (0,0,0) -- (0,1,0);
\draw[mid>] (0,0,1.5) -- (0,0,0);
\draw[mid>] (0,0,0) -- (0,0,-1.5);
}
\]

The plaquette term $B_p$ is the same operator as in the two-dimensional model. In particular, we may define this operator on each plaquette $p$ on each planar $\mathbb{Z}^2$ sublattice which has been oriented such that the edges point up and right.
There are actually two distinct orientations of embeddings of the $\mathbb{Z}^2$ lattice as a planar sublattice of $\mathcal{L}$ depending on which side of plane you look, but since $G$ is abelian, these two choices are actually equivalent for defining $B_p$.

The star term $A_s$ is the projection:
\[
A_s \coloneqq \sum_{ghk = \ell mn} 
\tikzmath{
\draw (0,0,1.5) -- (0,0,-1.5);
\draw (1,0,0) -- (-1,0,0);
\draw (0,1,0) -- (0,-1,0);

\filldraw (0,0,0.75) node[left]{$\scriptstyle P_h$} circle (.05cm);
\filldraw (0,0,-0.75) node[right]{$\scriptstyle P_m$} circle (.05cm);
\filldraw (0.65,0,0) node[below]{$\scriptstyle P_n$} circle (.05cm);
\filldraw (-0.5,0,0) node[above]{$\scriptstyle P_g$} circle (.05cm);

\filldraw (0,0.7,0) node[left]{$\scriptstyle P_\ell$} circle (.05cm);
\filldraw (0,-0.65,0) node[left]{$\scriptstyle P_k$} circle (.05cm);
}
\]
where $P_g$ is defined as in the 2D Quantum Double as the projection onto the state $|g\rangle$.

With these interaction operators, the Hamiltonian $H$ and the projection onto the local ground state space $p_\Lambda$ are defined as in Section \ref{subsec:qd}.

\begin{defn}
Given a character $\chi\colon G\to U(1)$,
  we define the operator $Z_\chi\colon \C[G] \to \C[G]$ by $Z_\chi|g\rangle=\chi(g)|g\rangle$.
\end{defn}

\begin{rem}
    The surrounding conditions given in Definition \ref{defn:surrounding-regions} can be defined analogously in the three-dimensional model using rectangular prisms instead of rectangles. 
    In particular, the incomplete $s$-surrounded condition $\Lambda \Subset_s \Delta$ is now modified in the sense that $\partial \Lambda \cap \partial \Delta = P$ is a non-empty two-dimensional plane in $\mathcal{L}$ that lies on one side of $\Lambda$ and $\Delta$. 
    As in the two-dimensional model, we consider $\widetilde{P}$ to be the set of edges comprising $P$ and the adjacent plane worth of edges to $P$ contained in $\Lambda$. 
    In this section, we will assume that $\widetilde{P}$ lies in the bottom face of $\Lambda$ and $\Delta$.
\end{rem}

Using the generators $R_g=L_g$ and $P_g$ for $M_{|G|}(\C)$, one can see that the three-dimensional analogue of the algorithm presented in \cite[Algorithm~3.10]{2307.12552} and an adaption of the argument presented in Proposition 3.9 of the same paper immediately yield the following theorem.

\begin{thm} \label{thm:LTO1,2,3DQuantumDouble}
    The axioms \ref{LTO:CompletelySurrounds}-\ref{LTO:Injectivity} hold for the 3D Quantum Double. 
    The boundary algebra $\mathfrak{C}(P)$ for the smooth cut is given by:
    \[
        \fC(P) = \rmC^\ast \{S^\chi_s, R^g_\ell | s \subset \tilde{P}, \ell \subset P,  g \in G\}
    \]
    where for every irreducible character $\chi$ of $G$, we define $S_s^\chi$ as the product of $Z_\chi$ operators shown in blue below:
    
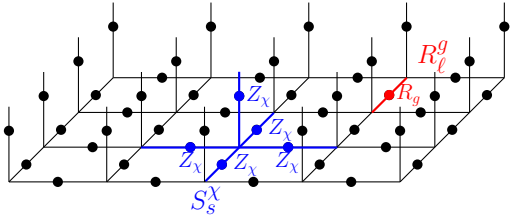
\[
        \begin{tikzpicture}[auto=center,every node/.style={circle, fill=black, scale=0.4}, scale=0.4]

            \def \dx{3.25};
            \def \dy{2.5};
            \def \dz{3};
            \def \nbx{5};
            \def \nby{2};
            \def \nbz{4};
            \def \nbxm{4};
            \def \nbym{1};
            \def \nbzm{3};
            
            \foreach \x in {1,...,\nbxm} {
                \foreach \y in {1,...,\nbym} {
                    \foreach \z in {1,...,\nbz} {
                        \node at (\dx/2+\x*\dx,\y*\dy,\z*\dz) { };
                    }
                }
            }
            
            \foreach \x in {1,...,\nbx} {
                \foreach \y in {1,...,\nbym} {
                    \foreach \z in {1,...,\nbzm} {
                        \node at (\x*\dx,\y*\dy,\dz/2+\z*\dz) { };
                    }
                }
            }

            \foreach \x in {1,...,\nbx} {
                \foreach \z in {1,...,\nbz}{
                    \draw (\x*\dx,\dy,\z*\dz) -- ( \x*\dx,\nby*\dy,\z*\dz);
                    \node at (\x*\dx,4.2,\z*\dz) { };
                }
            }
            
            \foreach \y in {1,...,\nbym} {
                \foreach \z in {1,...,\nbz}{
                    \draw (\dx,\y*\dy,\z*\dz) -- ( \nbx*\dx,\y*\dy,\z*\dz);
                }
            }

            \foreach \x in {1,...,\nbx} {
                \foreach \y in {1,...,\nbym}{
                    \draw (\x*\dx,\y*\dy,\dz) -- ( \x*\dx,\y*\dy,\nbz*\dz);
                }
            }
            
            \draw[blue, thick] (7.36, 0.1, 2.8) -- (7.36, 2.62, 2.8);
            \draw[blue, thick] (7.36, 0.1, 2.8) -- (10.62, 0.1, 2.8);
            \draw[blue, thick] (7.36, 0.1, 2.8) -- (4.10 , 0.1, 2.8);
            \draw[blue, thick] (7.36, 0.1, 2.8) -- (7.36, 0.1, 5.7);
            \draw[blue, thick] (7.36, 0.1, 2.8) -- (7.36, 0.1, -0.2);
            
            \node[fill=blue] at (7.36, 0.1, 1.28) {};
            \node[fill=blue] at (7.36, 0.11, 4.3) {};
            \node[fill=blue] at (7.36, 1.81, 2.8) {};
            \node[fill=blue] at (5.74, 0.11, 2.8) {};
            \node[fill=blue] at (8.99, 0.11, 2.8) {};
            
            \node[fill=none] at (8.95, -0.4, 2.8) {\huge{\textcolor{blue}{$Z_\chi$}}};
            \node[fill=none] at (5.74, -0.4, 2.8) {\huge{\textcolor{blue}{$Z_\chi$}}};
            \node[fill=none] at (8, 1.81, 2.8) {\huge{\textcolor{blue}{$Z_\chi$}}};
            \node[fill=none] at (8.15, 0.1,4.3) {\huge{\textcolor{blue}{$Z_\chi$}}};
            \node[fill=none] at (8.15, 0.1,1.28) {\huge{\textcolor{blue}{$Z_\chi$}}};
            
            \node[fill=none] at (8, 0.1,7.35) {\Huge{\textcolor{blue}{$S^\chi_s$}}};
            
            
            \draw[red,thick] (10.62,0.11,-0.2) -- (10.62,0.11,-3.2);
            \node[red] at (10.615, 0.115, -1.7) {};
            \node[fill=none] at (11.25, 0.115, -1.7) {\huge{\textcolor{red}{$R_g$}}};
            
            \node[fill=none] at (10.8, 0.115, -5) {\Huge{\textcolor{red}{$R^g_\ell$}}};
        \end{tikzpicture}
    \]
    
    Similarly, the boundary algebra for the rough cut is given by:
    \begin{align*}
        \mathfrak{C}&(P) = \rmC^\ast \{Z^{\chi}_{\ell}, Y^g_p | p \subset \tilde{P}, \ell \subset P \text{\normalfont{ internal}},g\in G\}
    \\&
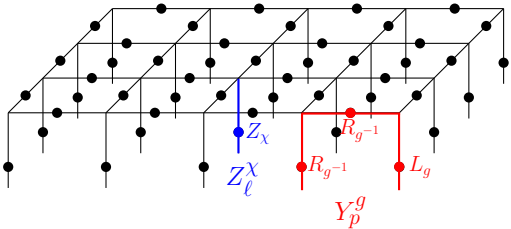
        \begin{tikzpicture}[auto=center,every node/.style={circle, fill=black, scale=0.4}, scale=0.4]
        \def \dx{3.25};
        \def \dy{2.5};
        \def \dz{3};
        \def \nbx{5};
        \def \nby{2};
        \def \nbz{4};
        \def \nbxm{4};
        \def \nbym{1};
        \def \nbzm{3};
        \foreach \x in {1,...,\nbxm} {
            \foreach \y in {1,...,\nbym} {
                \foreach \z in {1,...,\nbz} {
                    \node at (\dx/2+\x*\dx,\dy+\y*\dy,\z*\dz) { };
                }
            }
        }
        \foreach \x in {1,...,\nbx} {
            \foreach \y in {1,...,\nbym} {
                \foreach \z in {1,...,\nbzm} {
                    \node at (\x*\dx,\dy+\y*\dy,\dz/2+\z*\dz) { };
                }
            }
        }
        \foreach \x in {1,...,\nbx} {
            \foreach \z in {1,...,\nbz}{
                \draw (\x*\dx,\dy,\z*\dz) -- ( \x*\dx,\nby*\dy,\z*\dz);
                \node at (\x*\dx,3.2,\z*\dz) { };
            }
        }
        \foreach \y in {1,...,\nbym} {
            \foreach \z in {1,...,\nbz}{
                \draw (\dx,\dy+\y*\dy,\z*\dz) -- ( \nbx*\dx,\dy+\y*\dy,\z*\dz);
            }
        }
        \foreach \x in {1,...,\nbx} {
            \foreach \y in {1,...,\nbym}{
                \draw (\x*\dx,\dy+\y*\dy,\dz) -- ( \x*\dx,\dy+\y*\dy,\nbz*\dz);
            }
        }
        \draw[blue, thick] (7.36, 0.11, 2.8) -- (7.36, 2.59, 2.8);
        \node[fill=blue] at (7.363, 0.815, 2.8) {};
        \node[fill=none] at (8, 0.77, 2.8) {\huge{\textcolor{blue}{$Z_\chi$}}};
        \node[fill=none] at (7.5, -0.75,2.8) {\Huge{\textcolor{blue}{$Z^\chi_\ell$}}};
        \draw[red, thick] (10.57, 2.55, 5.7) -- (13.85, 2.55, 5.7);
        \draw[red, thick] (10.59, 2.55, 5.7) -- (10.59, 0, 5.7);
        \draw[red, thick] (13.815, 2.55, 5.7) -- (13.815, 0, 5.7);
        \node[fill=red] at (10.575, 0.778, 5.7) {};
        \node[fill=red] at (13.82, 0.78, 5.7) {};
        \node[fill=red] at (12.2, 2.57, 5.7) {};
        \node[fill=none] at (11.45, 0.77, 5.7) {\huge{\textcolor{red}{$R_{g^{-1}}$}}};
        \node[fill=none] at (14.515, 0.77, 5.7) {\huge{\textcolor{red}{$L_g$}}};
        \node[fill=none] at (12.5, 2, 5.7) {\huge{\textcolor{red}{$R_{g^{-1}}$}}};
        \node[fill=none] at (12.2, -0.75,5.7) {\Huge{\textcolor{red}{$Y^g_p$}}};
        \end{tikzpicture}
    \end{align*}
\end{thm}

Observe that the operators $Y^g_p$ above are truncations of the $B^{(g)}_p$ from \eqref{eq:QDBp}.
One can obtain an algebra isomorphic to $\fC(P)$ using the generators $P_g$ instead of $Z_\chi$; this was the generator used in Section \ref{sec: Boundary Algebra}. 
However, the use of irreducible characters allows us to easily obtain the abstract description of the algebra $\fC(P)$ as a direct sum of matrix algebras, as an abelian group is isomorphic to its dual group.
It is worth noting that since $G$ is abelian, the boundary algebra for smooth and rough cuts are always isomorphic. 

\begin{thm}
\label{thm:3DBoundaryAlgebraMatrixSummands}
    The boundary algebra $\fC(P)$ for both rough and smooth cuts has the following direct-sum decomposition:
    \[
    \fC(P) \cong \bigoplus^{|G|^{E-V}} M_{|G|^V}(\C).
    \]
    Here $E,V$ are the number of edges $\ell \subset P$ and the number of stars $s \subset \widetilde{P}$ in the smooth cut, respectively.
\end{thm}

\begin{proof}
    Since the boundary algebras for smooth and rough cuts are isomorphic, we only prove the statement for the boundary algebra for the smooth case. 
    As in the proof of Theorem \ref{thm:2d-qd-matrix-decomp}, we begin by noting that $\fC(P)$ is $\ast$-isomorphic to the following algebra:
    \begin{align*}
    &
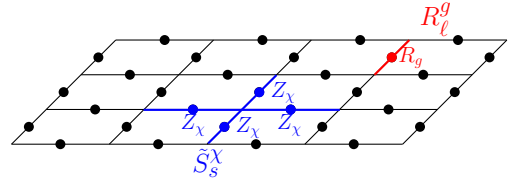
        \begin{tikzpicture}[auto=center,every node/.style={circle, fill=black, scale=0.4}, scale=0.4]
        \def \dx{3.25};
        \def \dy{2.5};
        \def \dz{3};
        \def \nbx{5};
        \def \nby{2};
        \def \nbz{4};
        \def \nbxm{4};
        \def \nbym{1};
        \def \nbzm{3};
        \foreach \x in {1,...,\nbxm} {
            \foreach \y in {1,...,\nbym} {
                \foreach \z in {1,...,\nbz} {
                    \node at (\dx/2+\x*\dx,\y*\dy,\z*\dz) { };
                }
            }
        }
        \foreach \x in {1,...,\nbx} {
            \foreach \y in {1,...,\nbym} {
                \foreach \z in {1,...,\nbzm} {
                    \node at (\x*\dx,\y*\dy,\dz/2+\z*\dz) { };
                }
            }
        }
        \foreach \y in {1,...,\nbym} {
            \foreach \z in {1,...,\nbz}{
                \draw (\dx,\y*\dy,\z*\dz) -- ( \nbx*\dx,\y*\dy,\z*\dz);
            }
        }
        \foreach \x in {1,...,\nbx} {
            \foreach \y in {1,...,\nbym}{
                \draw (\x*\dx,\y*\dy,\dz) -- ( \x*\dx,\y*\dy,\nbz*\dz);
            }
        }
        \draw[blue, thick] (7.36, 0.1, 2.8) -- (10.62, 0.1, 2.8);
        \draw[blue, thick] (7.36, 0.1, 2.8) -- (4.10 , 0.1, 2.8);
        \draw[blue, thick] (7.36, 0.1, 2.8) -- (7.36, 0.1, 5.7);
        \draw[blue, thick] (7.36, 0.1, 2.8) -- (7.36, 0.1, -0.2);
        \node[fill=blue] at (7.36, 0.1, 1.28) {};
        \node[fill=blue] at (7.36, 0.11, 4.3) {};
        \node[fill=blue] at (5.74, 0.11, 2.8) {};
        \node[fill=blue] at (8.99, 0.11, 2.8) {};
        \node[fill=none] at (8.95, -0.4, 2.8) {\huge{\textcolor{blue}{$Z_\chi$}}};
        \node[fill=none] at (5.74, -0.4, 2.8) {\huge{\textcolor{blue}{$Z_\chi$}}};
        \node[fill=none] at (8.15, 0.1,4.3) {\huge{\textcolor{blue}{$Z_\chi$}}};
        \node[fill=none] at (8.15, 0.1,1.28) {\huge{\textcolor{blue}{$Z_\chi$}}};
        \node[fill=none] at (8, 0.1,7.35) {\Huge{\textcolor{blue}{$\Tilde{S}^\chi_s$}}};
        \draw[red,thick] (10.62,0.11,-0.2) -- (10.62,0.11,-3.2);
        \node[red] at (10.615, 0.115, -1.7) {};
        \node[fill=none] at (11.25, 0.115, -1.7) {\huge{\textcolor{red}{$R_g$}}};
        \node[fill=none] at (10.8, 0.115, -5) {\Huge{\textcolor{red}{$R^g_\ell$}}};
        \end{tikzpicture}
\\
    \mathfrak{D}&(P) = \rmC^\ast \set{\widetilde{S}^\chi_s, R^g_\ell }{ s, \ell \subset P,  g \in G, \chi \text{  character}}.
    \end{align*}
    Therefore, it suffices to show the statement for $\mathfrak{D}(P)$. 
    Since $\fD(P)$ is manifestly a $*$-subalgebra of $M_{|G|}(\bbC)^{\otimes E}$,
    determining the direct sum decomposition of $\fD(P)$ is equivalent to characterizing the invariant subspaces of $\bbC[G]^{\otimes E}$ under the $\fD(P)$-action.
    
    Suppose that $G$ is a cyclic group of order $n$ for some $n \geq 1$, let $\omega_n$ be an $n$-th root of unity, and let $\chi$ be an irreducible character which generates the dual group $\widehat{G}$.
    It is straightforward to show that
    \[
    \beta_n = 
    \set{ \sum_{i=0}^{n-1} \omega_n^{m\cdot i} |r^i\rangle
    }
    {0 \leq m \leq n-1}
    \]
    is a basis for $\C[G]$ that diagonalizes the operator $R_g$ for all $g \in G$. 
    Therefore, $\beta_n^{\otimes E}$ is a basis for $\C[G]^{\otimes E}$ that diagonalizes $R^g_\ell$ for all $g \in G$ and all $\ell \subset P$. 
    Since $G$ is abelian, there are $|G|$ irreducible characters and each of these is a group homomorphism $G \to U(1)$, which necessarily takes vales in the powers of $\omega_n$.
    It follows that each $\widetilde{S}^\chi_s$ is an permutation matrix of order $n=|G|$ in the basis $\beta_n$. 
    Since the operators $\set{\widetilde{S}^\chi_s}{s\subset P}$ commute,
    and there are no additional relations amongst these operators,
    the size of each invariant subspace is $|G|^V$, 
    which immediately yields the stated direct-sum decomposition. 
    Since a general abelian group decomposes as a product of cyclics, this argument extends to the general case as well.
\end{proof}

\begin{rem}
Observe that the formula in Theorem \ref{thm:3DBoundaryAlgebraMatrixSummands} for the matrix decomposition of the boundary algebra also holds in the 2D setting for abelian groups.
There, $E-V$ is always 1, so we get $\bigoplus^{|G|} M_{|G|^V}(\bbC)$, which can be read directly off the Bratteli diagram \eqref{eq:BratteliHilbG}.
\end{rem}

As previously mentioned, one could apply the gauging procedure from Section \ref{sec:SPT Boundary Algebra} to compute the boundary algebra in the 3D model as well.


\newpage 
\bibliographystyle{alpha}
\onecolumngrid{
\footnotesize{
\bibliography{../../bibliography/bibliography}
}}

\end{document}